\documentclass[11pt]{article}
\usepackage{fullpage}
\usepackage[utf8]{inputenc} 
\usepackage{dirtytalk}
\usepackage{tikz,color,subcaption,url, comment}
\usetikzlibrary{matrix}
\usepackage{amsfonts,amssymb,amsmath,amsthm}
\theoremstyle{definition}

\newtheorem{theorem}{Theorem}[section]
\newtheorem{corollary}{Corollary}[theorem]
\newtheorem{lemma}[theorem]{Lemma}
\newtheorem{definition}{Definition}[section]

\newcommand{\toto}{xxx}

\usepackage{xargs}                      
\usepackage{xcolor}  
\usepackage[colorinlistoftodos,prependcaption,textsize=tiny]{todonotes}
\newcommandx{\ST}[2][1=]{\todo[linecolor=red,backgroundcolor=red!25,bordercolor=red,#1]{#2}}
\newcommandx{\AD}[2][1=]{\todo[linecolor=blue,backgroundcolor=blue!25,bordercolor=blue,#1]{#2}}
\newcommandx{\MP}[2][1=]{\todo[linecolor=green,backgroundcolor=green!25,bordercolor=green,#1]{#2}}
\newcommandx{\EA}[2][1=]{\todo[linecolor=orange,backgroundcolor=orange!25,bordercolor=orange,#1]{#2}}
\newcommandx{\RL}[2][1=]{\todo[linecolor=olive,backgroundcolor=olive!25,bordercolor=olive,#1]{#2}}

\usepackage[normalem]{ulem}
\newcommand\redout{\bgroup\markoverwith
	{\textcolor{red}{\rule[0.5ex]{2pt}{0.8pt}}}\ULon}

\makeatletter
\let\latex@@line\line
\def\line{\@ifnextchar(\latex@@line{\hbox to\hsize}}
\makeatother

\makeatletter
\newcommand{\settitle}{\@maketitle}
\makeatother

\title{Blockchain Abstract Data Type}

\author{Emmanuelle Anceaume$^\ddagger$, Antonella Del Pozzo$^{\star}$, Romaric Ludinard$^{\star\star} $, \\Maria Potop-Butucaru$^\dagger$, Sara Tucci-Piergiovanni$^\star$\\~\\
	$^\ddagger$CNRS, IRISA\\
	$^\star$CEA LIST, PC 174, Gif-sur-Yvette, 91191, France\\
	$^{\star\star}$ IMT Atlantique, IRISA \\
	$^\dagger$Sorbonne Université, CNRS, Laboratoire d'Informatique de Paris 6, LIP6, Paris, France \\
}
\date{}

\begin{document}
	
			\newcounter{linecounter}
			\newcommand{\linenumbering}{(\arabic{linecounter})}
			\renewcommand{\line}[1]{\refstepcounter{linecounter}
				\label{#1}
				\linenumbering}
			\newcommand{\resetline}{\setcounter{linecounter}{0}}
\maketitle

\begin{abstract}
 \normalsize
The presented work continues the line of recent distributed computing community efforts dedicated to the theoretical aspects of blockchains. This paper is the first to specify blockchains as a composition of \emph{abstract data types} all together with a hierarchy of \emph{consistency criteria} that formally characterizes the histories admissible for distributed  programs that use them. Our work is based on an original oracle-based construction that, along with new consistency definitions, captures the eventual convergence process in blockchain systems. The paper presents as well some results on implementability of the presented abstractions  and a mapping of representative existing blockchains from both academia and industry in our framework. 



\end{abstract}

\section{Introduction}\label{sec:intro}

The paper proposes a new data type to formally model blockchains and their behaviors. We aim at providing consistency criteria to capture the correct behavior of current blockchain proposals in a unified framework. It is already known that some blockchain implementations solve eventual consistency of an append-only queue using Consensus~\cite{AF2018, HyperLedger}. The question is about the consistency criterion of blockchains as Bitcoin~\cite{bitcoin} and Ethereum~\cite{Ethereum} that technically do not solve Consensus, and their relation with Consensus in general. 



We advocate that the key point to capture blockchain behaviors is to define consistency criteria allowing mutable operations to create forks and restricting the values read, i.e. modeling the data structure as \textit{an append-only tree} and not as an append-only queue.  This way we can easily define a semantics equivalent to  eventual consistent append-only queue but as well weaker semantics. More in detail, we define a semantic equivalent to  eventual consistent append-only queue by restricting any two reads to return two chains such that one is the prefix of the other. We call this consistency property Strong Prefix (already introduced in~\cite{GGGHS17}). Additionally, we define a weaker semantics restricting any two reads to return chains that have a divergent prefix for a finite interval of the history.	We call this consistency property Eventual Prefix.

%


Another peculiarity of blockchains lies in the notion of \textit{validity} of blocks, i.e. the blockchain must contain only blocks that satisfy a given predicate. Let us note that  validity can be achieved through proof-of-work (Dwork and Naor \cite{DworkN92}) or other agreement mechanisms.   We advocate that to abstract away implementation-specific validation mechanisms, the validation process must be encapsulated in an oracle model separated from the process of updating the data structure. Because the oracle is the only generator of valid blocks and only valid blocks can be appended, it follows that it is the oracle that grants the access to the data structure and it might also own a synchronization power to control the number of forks, in terms of branches of the tree from a given block. In this respect we define oracles models such that, depending on the model, the number of forks from a given block can be: unbounded, up to $k>1$, and $k=1$ (no fork) for the  strongest oracle model. 

The blockchain is then abstracted by an oracle-based construction in which the update and consistency of the tree data structure depends on the validation and synchronization power of the oracle.  

The main contribution of the paper is a formal unified framework providing blockchain consistency criteria that can be combined with  oracle models in a proper \textit{hierachy of abstract data types}~\cite{perrin:hal-01286755} independent of the  underlying communication and failure model. Thanks to the establishment of the formal framework the following implementability results are shown: 
 
	
	%
	%
	
	\begin{itemize}
	
	\item The strongest oracle, guaranteeing no fork, has Consensus number  $\infty$~ in the Consensus hierarchy of concurrent objects ~\cite{herlihy1991wait} (Theorem~\ref{t:consensusOracleReduction}). It must be noted that we considered Consensus  defined  in~\cite{redbelly17, gilad2017algorand,CachinKPS01}, in which the \emph{Validity} property states that  a valid block can be decided even if sent by a faulty process.
	
	\item The weakest oracle,  which validates a potential unbounded number of blocks to be appended to a given block, has Consensus number  $1$~  (Theorem~\ref{t:CTtoAS}). 
	

\item The impossibility to guarantee Strong Prefix in a message-passing system if forks are allowed (Theorem~\ref{t:noForkBeStrong}). 	This means that Strong Prefix needs the strongest oracle to be implemented, which is at least as strong as Consensus. 
	
\item  A necessary condition (Theorem~\ref{th:LRCNecessity}) for  Eventual Prefix in a message-passing system,  called \emph{Update Agreement} stating that  each update sent by a correct process must be eventually received by every correct process. The result implies that it is impossible to implement Eventual Prefix if even only one message sent by a correct process is dropped. 
	
	
\end{itemize}	
The proposed  framework along with the above-mentioned results helps in classifying existing blockchains in terms of their consistency and implementability. We used the framework to classify several blockchain proposals. We showed that Bitcoin~\cite{bitcoin} and Ethereum~\cite{Ethereum} have a validation mechanism that maps to our weakest oracle and then they only implement Eventual prefix, while other proposals maps to our strongest oracle, falling in the class of those that guarantee Strong Prefix (e.g. Hyperledger Fabric~\cite{HyperLedger}, PeerCensus~\cite{PeerCensus}, ByzCoin~~\cite{BizCoin}, see Section~\ref{sec:mapp-with-exist} for further details).



\noindent \textbf{Related Work.}
Formalisation of blockchains in the lens of distributed computing has been recognized as an extremely important topic \cite{Herlihy2017}. The topic is recent and to the best of our knowledge,  no other attempt proposed a unified  framework capturing both Consensus-based and proof-of-work blockchains, as the presented paper aims at proposing.

In \cite{AM2017}, the authors present a study about the relationship of  BFT consensus and blockchains. In order to abstract the proof-of-work mechanism the authors propose a specific oracle, in the same spirit of our oracle abstraction. While  their oracle is more specific then ours, since it makes a direct reference to proof-of-work properties, it offers as well a fairness property. Note that we do not formalize fairness properties in this paper, we only offer a generic merit parameter that can be used to define fairness. Let us note that apart from the fairness property, our oracle captures the semantics of \cite{AM2017}'s oracle. 


In parallel and independently of the work  in \cite{BTADT2018TR},  \cite{AF2018} proposes a  formalization of distributed ledgers modeled as  an ordered list of records. The authors propose in their formalization three consistency criteria: eventual consistency, sequential consistency and linearizability. They discuss how Hyperldger Fabric implements eventual consistency and propose implementations for sequential consistency and linearizability using a total order broadcast abstraction.  Interestingly, they show  that a distributed ledger that provides eventual consistency can be used to solve the consensus problem. These findings confirm our results about the necessity of Consensus to solve Strong Prefix and corroborate our mapping of Hyperledger Fabric. On the other hand the proposed formalization does not propose weaker consistency semantics more suitable for proof-of-work blockchains as BitCoin. Indeed, \cite{AF2018} continues and it is complementary to the work on the first formalisation of Bitcoin as a distributed ledger proposed in \cite{Anceaume17} where the distributed ledger is modelled as a simple register. These works suggest different abstractions  to model proof-of-work and Consensus-based blockchains, respectively.
The presented paper, on the other hand, thanks to our oracle-based construction  (not present in \cite{AF2018}, \cite{Anceaume17}) generalizes both  \cite{Anceaume17} and   \cite{AF2018} to encompass both kind of blockchains in a unified framework.

Finally,  \cite{ GGGHS17} presents an implementation of the  Monotonic Prefix Consistency (MPC) criterion  and showed that no criterion stronger than MPC can be implemented in a partition-prone message-passing system. Nicely, this result and more in general solvability results for eventual consistency \cite{Dubois2015} immediately apply to our Strong Prefix criterion.



\section{Preliminaries on shared object specifications based on Abstract Data Types}
\label{sec:paper-adt}
The basic idea underlying the use of abstract data types is to specify shared objects using two complementary facets~\cite{matthieu-book}:  a sequential specification that describes the semantics of the object, and a consistency criterion over concurrent histories, i.e. the set of admissible executions in a concurrent environment. In this work we are interested in consistency criteria achievable in a distributed environment  in which processes are sequential and communicate through message-passing. 

\subsection{Abstract Data Type (ADT)}
\label{sec:abstract-data-type}

The model used to specify an abstract data type is a form of transducer, as Mealy's machines, accepting an infinite but countable number of states. The values that can be taken by the data type are encoded in the abstract state, taken in a set $Z$. It is possible to access the object using the symbols of an input alphabet $A$. Unlike the methods of a class, the input symbols of the abstract data type do not have arguments. Indeed, as one authorizes a potentially infinite set of operations, the call of the same operation with different arguments is encoded by different symbols. An operation can have two types of effects. First, it can have a side-effect that changes the abstract state, the corresponding transition in the transition system being formalized by a transition function $\tau$. Second, operations can return values taken in an output alphabet $B$, which depend on the state in which they are called and an output function $\delta$. For example, the pop operation in a stack removes the element at the top of the stack (its side effect) and returns that element (its output).

The formal definition of abstract data types is as follows.


\begin{definition}{\bf (Abstract Data Type $T$)}
An abstract data type  is a 6-tuple $T=\langle A,B, Z, \xi_0, \tau, \delta \rangle $ where:
\begin{itemize} 
\item $A$ and $B$  are countable sets called input alphabet and output alphabet;
\item $Z$ is a countable set of abstract states  and $\xi_0$ is the initial abstract state; 
\item $\tau:Z \times A \rightarrow Z$ is the transition function;
\item $\delta:Z \times A \rightarrow B$ is the output function.
\end{itemize} 
\end{definition}

\begin{definition}{\bf (Operation)}
	Let $T=\langle A,B,Z, \xi_0, \tau, \delta \rangle$ be an abstract data type. An \emph{operation} of $T$ is an element of $\Sigma= A \cup (A \times B)$. We refer to a couple $(\alpha, \beta) \in A \times B$ as $\alpha / \beta$. We extend the transition function $\tau$ over the operations and apply $\tau$ on the operations input alphabet:

$$\tau_T : \begin{cases} Z \times \Sigma \rightarrow Z \\
						(\xi, \alpha) \mapsto \tau (\xi, \alpha) \text{ if } \alpha \in A \\
						(\xi, \alpha/ \beta) \mapsto \tau(\xi, \alpha) \text{ if } \alpha / \beta \in A\times B \\ 
						  \end{cases} 
$$	
	
\end{definition}

\subsection{Sequential specification of an ADT}
\label{sec:sequ-spec-an}

An abstract data type, by its transition system, defines the sequential specification of an object.
That is, if we consider a path that  traverses its system of transitions, then the word formed by the subsequent labels  on the path is part of the sequential specification of the abstract data type, i.e. it is a sequential history. 
The language recognized by an ADT is the set of all possible words. This language defines the sequential specification of the ADT. More formally, 

\begin{definition}{\bf(Sequential specification $L(T)$)}
A finite or infinite sequence 
$\sigma = {(\sigma_i)}_{i \in D} \in \Sigma^\infty$, $D = \mathbb{N}$ or  $D  = \{0, \dots, |\sigma|-1 \}$
is a \emph{sequential history} of an abstract data type $T$ if there exists a sequence of the same length $(\xi_{i+1})_{i \in D} \in Z ^\infty$ ($\xi_0$ has already been defined has the initial state) of states of $T$ such that, for any $i \in D$, 
	\begin{itemize}
		\item the output alphabet of $\sigma_i$ is compatible with $\xi_i$: $\xi_i \in \delta_T ^{-1}(\sigma_i)$;
		\item the execution of the operation $\sigma_i$ is such that the state changed from $\xi_i$ to $\xi_{i+1}$: $\tau_T(\xi_i,\sigma_i)=\xi_{i+1}$.
	\end{itemize}
	The \emph{sequential specification} of $T$ is the set of all its possible sequential histories $L(T)$. 
\end{definition}

\subsection{Concurrent  histories of an ADT}
\label{sec:conc-spec-adts}

Concurrent histories are defined considering asymmetric event structures, i.e., partial order relations among events executed by different processes~\cite{matthieu-book}.
\begin{definition}
\label{def:adthistory}{\bf(Concurrent history $H$)}
	The execution of a program that uses an abstract data type  T =$\langle$ A, B, Z, $\xi_0, \tau, \delta \rangle $ defines a concurrent history  $H=\langle \Sigma, E, \Lambda, \mapsto, \prec, \nearrow \rangle$,  where
	\begin{itemize}
		\item $\Sigma = A \cup ( A \times B)$ is a countable set of operations;
		\item $E$ is a countable set of events that contains all the ADT operations invocations  and all ADT operation response events;
		\item $\Lambda:E \rightarrow \Sigma$ is a function which associates events to the operations in $\Sigma$;
		\item $\mapsto$: is the process order relation over the events in $E$. Two events are ordered by $\mapsto$ if they are produced by the same process;
		\item  $\prec$: is the operation order, irreflexive order over the events of $E$. For each couple $(e, e') \in E^2$, if $e$ is an operation invocation and $e'$ is the response for the same operation then $e \prec e'$, if $e'$ is the invocation of an operation occurred at time $t'$ and $e$ is the response of another operation occurred at time $t$ with $t < t'$ then $e \prec e'$; 
		\item $\nearrow$: is the program order, irreflexive order over $E$, for each couple $(e, e') \in E^2$ with $e \neq e'$ if $e \mapsto e'$ or $e \prec e'$ then $e \nearrow e'$.	
\end{itemize}
\end{definition}

\subsection{Consistency criterion}
\label{sec:coher-crit-conc}

The consistency criterion characterizes which concurrent histories are admissible for a given abstract data type. It  can be viewed as a function that associates a concurrent specification to abstract data types. Specifically,

\begin{definition}{\bf(Consistency criterion $C$)}
	A consistency criterion is a function $$ C: \mathcal{T} \rightarrow \mathcal{P}(\mathcal{H})$$ where $\mathcal{T}$ is the set of abstract data types, $ \mathcal{H}$ is a set of histories and  $ \mathcal{P}(\mathcal{H})$ is the sets of parts of $ \mathcal{H}$. 
\end{definition}

 Let $\mathcal{C}$ be the set of all the consistency criteria.
An algorithm $A_T$ implementing the ADT $T \in \mathcal{T}$ is $C$-consistent  with respect to criterion $C \in \mathcal{C}$  if all the operations terminate and all the admissible executions are $C$-consistent, i.e. they belong to the set of histories $C(T)$.


\section{BlockTree and Token oracle ADTs}
\label{sec:blocktr-token-oracle}
In this section we present the BlockTree and the token Oracle ADTs  along with consistency criteria. 

\subsection{BlockTree ADT}
\label{subsec:blocktree}
We formalize the data structure implemented by blockchain-like systems as a \emph{directed rooted tree} $bt=(V_{bt},E_{bt})$ called \emph{BlockTree}. Each vertex of the BlockTree is a \emph{block} and any edge points backward to the root, called \emph{genesis block}. The height of a block refers to its distance to the root. We denote by \(b_k\) a block located at height \(k\). By convention, the root of the BlockTree is denoted by \(b_0\). Blocks are said valid if  they satisfy a predicate $P$ which is application dependent (for instance, in Bitcoin, a block is considered valid if it can be connected to the current blockchain and does not contain transactions that double spend a previous transaction). 
We represent by $\mathcal{B}$  a countable and non empty set of blocks and by $\mathcal{B'}\subseteq \mathcal{B}$ a countable and non empty set of valid blocks, i.e.,  $\forall b \in \mathcal{B'}$, $P(b)=\top$.  By assumption $b_0 \in \mathcal{B'}$; We also denote by $\mathcal{BC}$ a countable non empty set of blockchains, where a blockchain is a path from a leaf of $bt$ to $b_0$. A blockchain is denoted by $bc$. Finally, $\mathcal{F}$ is a countable non empty set of selection functions, $f \in \mathcal{F} : \mathcal{BT} \rightarrow \mathcal{BC}$;  $f(bt)$ selects a blockchain $bc$ from the BlockTree $bt$ (note that $b_0$ is not returned) and if $bt=b_0$ then $f(b_0)=b_0$.  This reflects for instance the longest chain or the heaviest chain used in some blockchain implementations. The selection function $f$ and the predicate $P$ are parameters of the ADT which are encoded in the state and do not change over the computation.

The following notations are also deeply used: $\{b_0\}^\frown {f(bt)}$ represents the concatenation of $b_0$ with the blockchain of $bt$; and  $\{b_0\}^\frown {f(bt)}^ \frown \{b\}$ represents  the concatenation of $b_0$ with the blockchain of $bt$  and a block $b$;

\subsubsection{Sequential specification of the BlockTree}
\label{sec:blocktr-sequ-spec}

The sequential specification of the BlockTree is defined as follows.

\begin{definition}[BlockTree ADT ($BT$-$ADT$)]

The BlockTree Abstract Data Type 
is the 6-tuple 
BT-ADT=$\langle A=\{ {\sf append}$($b$)$, {\sf read}$()$: 
b \in \mathcal{B}\}, B= \mathcal{BC}\cup\{ {\sf true},{\sf false}\}, Z= \mathcal{BT} \times \mathcal{F} \times (\mathcal{B} \rightarrow \{ {\sf true},{\sf false}\}) $, $\xi_0 = (bt^0, f), \tau, \delta \rangle$, where the transition function $\tau:Z \times A \rightarrow Z$ is defined by
\begin{itemize}
	\item  $\tau((bt,f, P), {\sf append}(b))= (\{b_0\}^\frown {f(bt)}^ \frown \{b\},f,P)$  if $b\in \mathcal{B'}$; $(bt,f, P)$ otherwise;
	\item  $\tau((bt,f, P), {\sf read}())= (bt,f, P)$, 
\end{itemize}
and the output function $\delta:Z \times A \rightarrow B$ is defined by
\begin{itemize}
	\item  $\delta((bt,f, P), {\sf append}(b))= {\sf true}$ if $b \in \mathcal{B'}$; ${\sf false}$ otherwise;
	\item  $\delta((bt,f, P), {\sf read}())= \{b_0\}^\frown f(bt)$;
	\item  $\delta((bt_0,f, P), {\sf read}())= b_0$.
\end{itemize}


%
%
%
%
%

\end{definition}

The semantic of the {\sf read} and the {\sf append} operations directly depend on the selection function $f \in \mathcal{F}$. In this work we let this function generic to suit the different blockchain implementations. In the same way, predicate $P$ is let unspecified. The predicate $P$ mainly abstracts the creation  process of a block, which may fail or successfully terminate. This process will be further specified in Section \ref{sec:oracle}. 

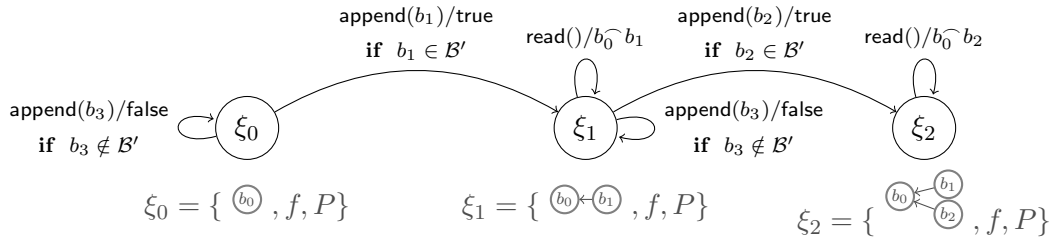
\begin{figure}
	\begin{tikzpicture}[->, auto]
	\tikzset{stato/.style={draw,circle, node distance=4.5cm}}
	\tikzset{dettaglio/.style={black!70}}
	\tikzset{blocco/.style={
			circle,
			inner sep=1pt,
			thick,
			align=center,
			draw=gray,
		}
	}
	
		\node[stato] (0) at (0,0) {$\xi_0 $};
		
		\node[dettaglio, below of=0] (s0) {$\xi_0 = \{$
			\begin{tikzpicture}[scale=0.8, every node/.style={scale=0.8}]	
				\node  [blocco] (i11) at (0.5,3)  {\scriptsize $b_0$};	
			\end{tikzpicture}
			$, f,P\}$};

		\node[stato, right of=0] (1) {$\xi_1$}; 
		\node[dettaglio, below of=1] (s1) {$\xi_1= \{$			
			\begin{tikzpicture}[scale=0.8, every node/.style={scale=0.8}]	
				\node  [blocco] (i11) at (0.5,3)  {\scriptsize $b_0$};
				\node  [blocco]  (i12) at (1.2,3) {\scriptsize $b_1$};
				\draw[<-] (i11) edge (i12)  ;
			\end{tikzpicture}
			$, f ,P\}$};

		\path (0) edge [bend left] node[align=center] {\scriptsize {\sf append}$ (b_1)/ {\sf true}$ \\ {\scriptsize{\bf if } $b_1 \in \mathcal{B'} $}} (1);
		\path (0) edge [loop left] node[align=center] {\scriptsize {\sf append}$( b_3)/ {\sf false}$ \\ {\scriptsize{\bf if } $b_3 \notin \mathcal{B'} $}} (0);

		\node[stato, right of=1] (2) {$\xi_2$}; 
		\node[dettaglio, below of=2] (21) {$\xi_2= \{$
			\begin{tikzpicture}[scale=0.8, every node/.style={scale=0.8}]
				\node  [blocco] (i11) at (0.5,3)  {\scriptsize $b_0$};
				\node  [blocco]  (i12) at (1.3,3.2) {\scriptsize $b_1$};
				\node [blocco] (i13) at (1.3,2.7) {\scriptsize $b_2$};
				\draw[<-] (i11) edge (i12) (i11) edge (i13) ;
			\end{tikzpicture}
			$, f ,P\}$};	
		
		\path (1) edge [bend left] node[align=center] {\scriptsize {\sf append}$(b_2)/ {\sf true}$ \\ {\scriptsize{\bf if } $b_2 \in \mathcal{B'} $}} (2);
		\path (1) edge [loop right] node[align=center] {\scriptsize {\sf append}$(b_3)/ {\sf false}$ \\ {\scriptsize{\bf if } $b_3 \notin \mathcal{B'} $}} (1);

		\path (1) edge [loop above] node[align=center] {\scriptsize {\sf read}$()/ {b_0 ^\frown b_1}$} (1);	
		\path (2) edge [loop above] node[align=center] {\scriptsize {\sf read}$()/ {b_0 ^\frown b_2}$} (2);		
	\end{tikzpicture}
\caption{A possible path of the transition system defined  by the BT-ADT. We use the following syntax on the edges: {\sf operation}/{\sf output}.}\label{fig:BT-ADT}
\end{figure}

\subsubsection{Concurrent specification of a BT-ADT and consistency criteria}
\label{sec:conc-spec-bt}
The concurrent specification of the BT-ADT is the set of concurrent histories. A $BT$-$ADT$ consistency  criterion is a function that returns the set of concurrent histories admissible for a BlockTree abstract data type.  We define two $BT$ consistency criteria:  \emph{BT Strong consistency} and \emph{BT Eventual consistency}. For ease of readability, we employ the following notations:
\begin{itemize}
	\item $E(a^*,r^*)$ is an infinite set containing an infinite number of {\sf append}$()$ and {\sf read}$()$ invocation and response events;
	\item $E(a,r^*)$ is an infinite set containing {\em (i)} a finite number of {\sf append}$()$ invocation and response events and {\em (ii)} an infinite number of {\sf read}$()$ invocation and response events;
	\item  $e_{inv}(o)$ and $e_{rsp}(o)$ indicate respectively the invocation and response event of an operation $o$; and $e_{rsp}(r):bc$ denotes the returned blockchain $bc$ associated with the  response event $e_{rsp}(r)$;
	\item ${\sf score}: \mathcal{BC}\rightarrow \mathbb{N}$ denotes a monotonic increasing deterministic function that takes as input a blockchain $bc$ and returns a natural number $s$ as score of $bc$, which can be the height, the weight, etc. Informally we refer to such value as the score of a blockchain; by convention we refer to the score of the blockchain uniquely composed by the genesis block as $s_0$, i.e. ${\sf score}(\{b_0\})=s_0$. Increasing monotonicity means that ${\sf score}(bc^\frown \{b\}) > {\sf score}(bc)$; 
	\item ${\sf mcps}: \mathcal{BC} \times \mathcal{BC} \rightarrow \mathbb{N}$ is a function that given two blockchains $bc$ and $bc'$ returns the score of the maximal common prefix between $bc$ and $bc'$;
	\item $bc \sqsubseteq bc^\prime$ iff $bc$ prefixes $bc^\prime$.
\end{itemize}

\paragraph{BT Strong consistency.}

The BT Strong Consistency criterion is the conjunction of the following four properties. The block validity property imposes that each block in a blockchain returned by a {\sf read}$()$ operation is \emph{valid} (i.e., satisfies predicate $P$) and has been inserted in the BlockTree with the {\sf append}$()$ operation. The Local monotonic read states  that, given the sequence of {\sf read}$()$ operations at the same process, the score of the returned blockchain never decreases.  The Strong prefix property states that for each couple of read operations, one of the returned blockchains is a prefix of the other returned one (i.e., the prefix never diverges).  Finally, the Ever growing tree states that scores of returned blockchains eventually grow. More precisely, let $s$ be the score of the blockchain returned by a read response event $r$ in $E(a^*,r^*)$, then for each {\sf read}$()$ operation $r$, the set of {\sf read}$()$ operations $r'$ such that $e_{rsp}(r)\nearrow e_{inv}(r')$ that do not return blockchains with a score greater than $s$ is finite.  More formally, the BT Strong consistency criterion is defined as follows:


\begin{definition}[BT Strong  Consistency criterion ($SC$)]
	A concurrent history $H=\langle \Sigma, E, \Lambda, \mapsto, \prec, \nearrow \rangle$ of the system that uses a BT-ADT verifies the BT Strong Consistency criterion if the following  properties hold:
	\begin{itemize}
		\item {\bf Block validity:} 
		 $\forall e_{rsp}(r) \in E, \forall b \in e_{rsp}(r):bc, b \in \mathcal B' \wedge \exists e_{inv}({\sf append}(b)) \in E,$ 
		 
		 $\;\;\;\;\;\;\;\;\;\;\;\;\;\;\;\;\;\;\;\;\;\;\;\;\;\;\;\;\;\;\;\;\ e_{inv}(append(b)) \nearrow e_{rsp}(r).$	
		\item{\bf Local monotonic read:} 
		$$\forall e_{rsp}(r), e_{rsp}(r') \in E^2, \text{ if } e_{rsp}(r) \mapsto e_{inv}(r'), \text{ then } {\sf score}(e_{rsp}(r):bc)\leq {\sf score}(e_{rsp}(r'):bc').$$
		\item {\bf Strong prefix: } 
		$$\forall e_{rsp}(r), e_{rsp}(r') \in E^2,  (e_{rsp}(r'):bc' \sqsubseteq e_{rsp}(r):bc) \vee (e_{rsp}(r):bc \sqsubseteq e_{rsp}(r'):bc').$$		
		\item {\bf Ever growing tree:} 
$\forall e_{rsp}(r)  \in E(a^*,r^*), s={\sf score}(e_{rsp}(r):bc)\text{ then} $

		  $\;\;\;\;\;\;\;\;\;\;\;\;\;\;\;\;\;\;\;\;\;\;\;\;\;\;\;\;\;\;\;\;\ |\{ e_{inv}(r') \in  E \mid e_{rsp}(r) \nearrow e_{inv}(r'), {\sf score}(e_{rsp}(r'):bc')\leq s \}| <\infty.$
	\end{itemize}
\end{definition}

Figure \ref{fig:StrongHistory} shows a concurrent history $H$ admissible by the BT Strong consistency criterion. In this example the score is the length $l$ of the blockchain and the selection function $f$ selects the longest blockchain, and in case of equality, selects the largest based on the lexicographical order. For ease of readability, 
we do not depict the {\sf append}$()$ operation. We assume the block validity property is satisfied. The Local monotonic read is easily verifiable as for each couple of read blockchains one prefixes the other. The first {\sf read}$()$ $r$ operation, enclosed in a black rectangle, is taken as reference to check the consistency criterion (the criterion has to be iteratively verified for each {\sf read}$()$ operation). Let $l$ be the score of the blockchain returned by $r$. We can identify two sets, enclosed in rectangles defined by different patterns: \emph{(i)} the finite sets of {\sf read}$()$ operations such that the score associated to each blockchain returned is smaller than or equal to $l$, and \emph{(ii)} the infinite set of {\sf read}$()$ operations such that the score is greater than $l$. We can iterate the same reasoning for each {\sf read}$()$ operation in $H$. Thus $H$ satisfies the Ever growing tree property.

\tikzset{
	mystyle/.style={
		circle,
		inner sep=1pt,
		thick,
		align=center,
		draw=black,
	}
}

\tikzset{
	mystyleLight/.style={
		circle,
		inner sep=1pt,
		thick,
		align=center,
		draw=gray!40,
	    text=gray!40
	}
}

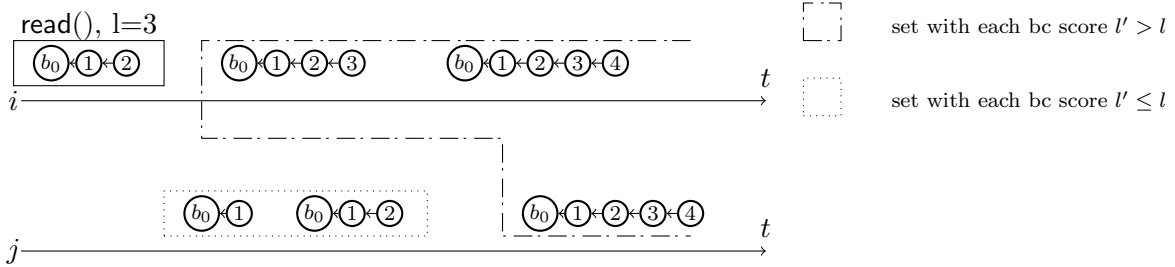
\begin{figure}[h]
	\begin{tikzpicture}
		\node  [mystyle] (i11) at (0.5,3)  {\scriptsize $b_0$};
		\node  [mystyle]  (i12) at (1,3) {\scriptsize $1$};
		\node [mystyle] (i13) at (1.5,3) {\scriptsize $2$};
			\draw[<-] (i11) edge (i12) (i12) edge (i13) ;

		\node [mystyle] (i21) at (3,3) {\scriptsize $b_0$};
		\node [mystyle] (i22) at (3.5,3) {\scriptsize $1$};
		\node [mystyle] (i23) at (4,3) {\scriptsize $2$};
		\node [mystyle] (i24) at (4.5,3) {\scriptsize $3$};
			\draw[<-] (i21) edge (i22) (i22) edge (i23) (i23) edge (i24) ;		
		
		\node [mystyle] (i31) at (6,3) {\scriptsize $b_0$};
		\node [mystyle] (i32) at (6.5,3) {\scriptsize $1$};
		\node [mystyle] (i33) at (7,3) {\scriptsize $2$};
		\node [mystyle] (i34) at (7.5,3) {\scriptsize $3$};
		\node [mystyle] (i35) at (8,3) {\scriptsize $4$};
			\draw[<-] (i31) edge (i32) (i32) edge (i33) (i33) edge (i34) (i34) edge (i35) ;

	\node[]() at (0,2.5){$i$}; \draw[->] (.1,2.5) -- (10,2.5);	\node[]() at (10,2.8){$t$};	

		\node [mystyle] (j11) at (2.5,1) {\scriptsize $b_0$};
		\node [mystyle] (j12) at (3,1) {\scriptsize $1$};
			\draw[<-] (j11) edge (j12) 
			;
			
		\node [mystyle] (j21) at (4,1) {\scriptsize $b_0$};
		\node [mystyle] (j22) at (4.5,1) {\scriptsize $1$};
		\node [mystyle] (j23) at (5,1) {\scriptsize $2$};
			\draw[<-] (j21) edge (j22) (j22) edge (j23) 
			;		
		
		\node [mystyle] (j31) at (7,1) {\scriptsize $b_0$};
		\node [mystyle] (j32) at (7.5,1) {\scriptsize $1$};
		\node [mystyle] (j33) at (8,1) {\scriptsize $2$};
		\node [mystyle] (j34) at (8.5,1) {\scriptsize $3$};
		\node [mystyle] (j35) at (9,1) {\scriptsize $4$};
			\draw[<-] (j31) edge (j32) (j32) edge (j33) (j33) edge (j34) (j34) edge (j35) ;

		\node[]() at (0,0.5){$j$}; \draw[->] (.1,0.5) -- (10,0.5); \node[]() at (10,.8){$t$};	
		
		\draw[draw=black] (0,2.7) rectangle (2,3.3); \node[]() at (1,3.5){{\sf read}$()$, l=3};
		
		\draw[ dash pattern={on 7pt off 2pt on 1pt off 3pt}] (9,3.3) -- (2.5,3.3) -- (2.5,2) -- (6.5,2) -- (6.5, 0.7) -- (9, 0.7) ;

		\draw[ dotted] (2,0.7) rectangle (5.5,1.3);  
		
		\draw[dotted] (10.5, 2.3) rectangle (11,2.8);
		\node[]() at (13.5,2.5){\scriptsize set with each bc score $l' \leq l$};	
		\draw[dash pattern={on 7pt off 2pt on 1pt off 3pt}] (10.5, 3.3) rectangle (11,3.8); 
		\node[]() at (13.5,3.5){\scriptsize set with each bc score $l'>l$};
	\end{tikzpicture}
	\caption{Concurrent history that satisfies the  BT Strong consistency criterion. In such scenario $f$ selects the longest blockchain and the blockchain score is length $l$.}\label{fig:StrongHistory}
\end{figure}

\paragraph{BT Eventual consistency.}
\label{sec:bt-event-cons}

The BT Eventual consistency criterion is the conjunction of the block validity, 
the Local monotonic read and the Ever growing tree of the BT Strong consistency criterion together with the Eventual prefix which states that for each blockchain returned by a {\sf read}$()$ operation with $s$ as score, then eventually all the {\sf read}$()$ operations will return blockchains sharing the same maximum common prefix at least up to $s$. Say differently, let $H$ be a history with an infinite number of {\sf read}$()$ operations, and let $s$ be the score of the blockchain returned by a read $r$, then the set of {\sf read}$()$ operations $r'$, such that $e_{rsp}(r)\nearrow e_{inv}(r')$, that do not return blockchains sharing the same prefix at least up to $s$ is finite.

\begin{definition}[Eventual prefix property]\label{def:eventual=preficx}
Given a concurrent history $H=\langle \Sigma, E(a,r^*), \Lambda, \mapsto, \prec, \nearrow \rangle$ of the system that uses a BT-ADT, we denote by \(s\), for any {\sf read} operation \(r \in \Sigma\) such that \(\exists e \in E(a,r^*), \Lambda(r) = e\), the score of the returned blockchain, \emph{i.e.}, \(s={\sf score}(e_{rsp}(r):bc)\).
We denote by \(E_{r}\) the set of response events of read operations that occurred after \(r\) response, \emph{i.e.} \(E_{r} = \{ e \in E \mid \exists r^\prime \in \Sigma, r^\prime={\sf read}, e=e_{rsp}(r^\prime) \wedge e_{rsp}(r) \nearrow e_{rsp}(r')  \}\). Then, H satisfies the Eventual prefix property if for all {\sf read()} operations \(r \in \Sigma\)  with score $s$,  \[
| \{ (e_{rsp}(r_h), e_{rsp}(r_k)) \in E_{r}^2 | h\neq k, {\sf mpcs}(e_{rsp}(r_h):bc_h, e_{rsp}(r_k):bc_k) < s \} | < \infty
\]
\end{definition}

 The Eventual prefix properties captures the fact that two or more concurrent blockchains can co-exist in a finite interval of time, but that  ly all the participants adopts a same branch for each cut of the history. This cut of the history is defined by a read that picks up a blockchain with a given score.  

Based on this definition, the BT Eventual consistency criterion is defined as follows:

\begin{definition}[BT Eventual consistency criterion  $EC$]
	A concurrent history $H=\langle \Sigma, E, \Lambda, \mapsto, \prec, \nearrow \rangle$ of the system that uses a BT-ADT verifies the \emph{BT Eventual consistency criterion} if it satisfies the Block validity, Local monotonic read, Ever growing tree, and the Eventual prefix properties.
\end{definition}


		

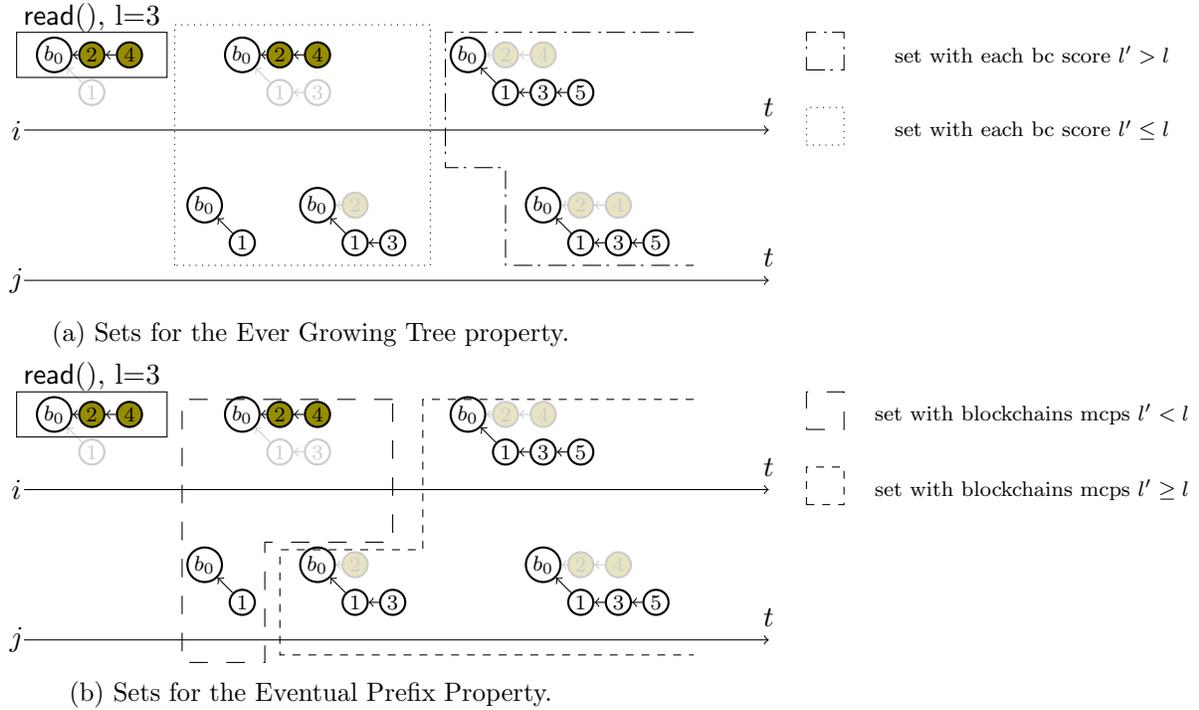
\begin{figure}[h]
	\begin{subfigure}{.5\textwidth}
	\begin{tikzpicture}
	\node [mystyle] (i11) at (0.5,3.5) {\scriptsize $b_0$};
	\node [mystyleLight] (i14) at (1,3) {\scriptsize $1$};
	\node [mystyle, fill=olive] (i12) at (1,3.5) {\scriptsize $2$};
	\node [mystyle, fill=olive] (i13) at (1.5,3.5) {\scriptsize $4$};
	\draw[<-] (i11) edge (i12) (i12) edge (i13) ;
	\draw[<-, gray!40] (i11)edge(i14);

	\node [mystyle] (i21) at (3,3.5) {\scriptsize $b_0$};
	\node [mystyle, fill=olive] (i22) at (3.5,3.5) {\scriptsize $2$};
	\node [mystyle, fill=olive] (i23) at (4,3.5) {\scriptsize $4$};
	\node [mystyleLight] (i24) at (3.5,3) {\scriptsize $1$};
	\node [mystyleLight] (i25) at (4,3) {\scriptsize $3$};
	\draw[<-] (i21) edge (i22) (i22) edge (i23);	
	\draw[<-, gray!40] (i21) edge  (i24) (i24)edge(i25);	
	
	\node [mystyle] (i31) at (6,3.5) {\scriptsize $b_0$};
	\node [mystyleLight, fill=olive!20] (i35) at (6.5,3.5) {\scriptsize $2$};
	\node [mystyleLight, fill=olive!20] (i36) at (7,3.5) {\scriptsize $4$};
	\node [mystyle] (i32) at (6.5,3) {\scriptsize $1$};
	\node [mystyle] (i33) at (7,3) {\scriptsize $3$};
	\node [mystyle] (i34) at (7.5,3) {\scriptsize $5$};
	\draw[<-] (i31) edge (i32) (i32) edge (i33) (i33) edge (i34) ;
	\draw[<-, gray!40] (i31) edge  (i35) (i35) edge(i36);
		
	\node[]() at (0,2.5){$i$}; \draw[->] (.1,2.5) -- (10,2.5);	\node[]() at (10,2.8){$t$};
	
	\node [mystyle] (j11) at (2.5,1.5) {\scriptsize $b_0$};
	\node [mystyle] (j12) at (3,1) {\scriptsize $1$};
	\draw[<-] (j11) edge (j12) 
	;
	
	\node [mystyle] (j21) at (4,1.5) {\scriptsize $b_0$};
	\node [mystyle] (j22) at (4.5,1) {\scriptsize $1$};
	\node [mystyle] (j23) at (5,1) {\scriptsize $3$};
	\node [mystyleLight, fill=olive!20] (j24) at (4.5,1.5) {\scriptsize $2$};
	\draw[<-] (j21) edge (j22) (j22) edge (j23);	
	\draw[<-, gray!40] (j21) edge (j24);	
	
	\node [mystyle] (j31) at (7,1.5) {\scriptsize $b_0$};
	\node [mystyle] (j32) at (7.5,1) {\scriptsize $1$};
	\node [mystyle] (j33) at (8,1) {\scriptsize $3$};
	\node [mystyle] (j34) at (8.5,1) {\scriptsize $5$};
	\node [mystyleLight, fill=olive!20] (j35) at (7.5,1.5) {\scriptsize $2$};
	\node [mystyleLight, fill=olive!20] (j36) at (8,1.5) {\scriptsize $4$};
	\draw[<-] (j31) edge (j32) (j32) edge (j33) (j33) edge (j34) ;	
	\draw[<-, gray!40] (j31) edge (j35) (j35) edge (j36);

		\node[]() at (0,0.5){$j$}; \draw[->] (.1,0.5) -- (10,0.5); \node[]() at (10,.8){$t$};	
	
	\draw[draw=black] (0,3.2) rectangle (2,3.8); \node[]() at (1,4){{\sf read}$()$, l=3};
	
	\draw[dash pattern={on 7pt off 2pt on 1pt off 3pt}] (9,3.8) -- (5.7,3.8) -- (5.7,2) -- (6.5,2) -- (6.5, 0.7) -- (9, 0.7) ; 
	
	\draw[dotted] (2.1, 0.7) -- (2.1,3.9) --(5.5,3.9) -- (5.5, 0.7) -- (2.1, 0.7);
	
%
	
	\draw[dotted] (10.5, 2.3) rectangle (11,2.8);
	\node[]() at (13.5,2.5){\scriptsize set with each bc score $l' \leq l$};	
	\draw[dash pattern={on 7pt off 2pt on 1pt off 3pt}] (10.5, 3.3) rectangle (11,3.8);
	\node[]() at (13.5,3.5){\scriptsize set with each bc score $l'>l$};
	\end{tikzpicture}
	\caption{Sets for the Ever Growing Tree property.} \label{fig:EventualHistoryEGT}
	\end{subfigure}
\\
	\begin{subfigure}[ASD]{.5\textwidth}
	\begin{tikzpicture}
	\node [mystyle] (i11) at (0.5,3.5) {\scriptsize $b_0$};
	\node [mystyleLight] (i14) at (1,3) {\scriptsize $1$};
	\node [mystyle, fill=olive] (i12) at (1,3.5) {\scriptsize $2$};
	\node [mystyle, fill=olive] (i13) at (1.5,3.5) {\scriptsize $4$};
	\draw[<-] (i11) edge (i12) (i12) edge (i13) ;
	\draw[<-, gray!40] (i11)edge(i14);

	\node [mystyle] (i21) at (3,3.5) {\scriptsize $b_0$};
	\node [mystyle, fill=olive] (i22) at (3.5,3.5) {\scriptsize $2$};
	\node [mystyle, fill=olive] (i23) at (4,3.5) {\scriptsize $4$};
	\node [mystyleLight] (i24) at (3.5,3) {\scriptsize $1$};
	\node [mystyleLight] (i25) at (4,3) {\scriptsize $3$};
	\draw[<-] (i21) edge (i22) (i22) edge (i23);	
	\draw[<-, gray!40] (i21) edge  (i24) (i24)edge(i25);	
	
	\node [mystyle] (i31) at (6,3.5) {\scriptsize $b_0$};
	\node [mystyleLight, fill=olive!20] (i35) at (6.5,3.5) {\scriptsize $2$};
	\node [mystyleLight, fill=olive!20] (i36) at (7,3.5) {\scriptsize $4$};
	\node [mystyle] (i32) at (6.5,3) {\scriptsize $1$};
	\node [mystyle] (i33) at (7,3) {\scriptsize $3$};
	\node [mystyle] (i34) at (7.5,3) {\scriptsize $5$};
	\draw[<-] (i31) edge (i32) (i32) edge (i33) (i33) edge (i34) ;
	\draw[<-, gray!40] (i31) edge  (i35) (i35) edge(i36);
	
	\node[]() at (0,2.5){$i$}; \draw[->] (.1,2.5) -- (10,2.5);	\node[]() at (10,2.8){$t$};
	
	\node [mystyle] (j11) at (2.5,1.5) {\scriptsize $b_0$};
	\node [mystyle] (j12) at (3,1) {\scriptsize $1$};
	\draw[<-] (j11) edge (j12) 
	;
	
	\node [mystyle] (j21) at (4,1.5) {\scriptsize $b_0$};
	\node [mystyle] (j22) at (4.5,1) {\scriptsize $1$};
	\node [mystyle] (j23) at (5,1) {\scriptsize $3$};
	\node [mystyleLight, fill=olive!20] (j24) at (4.5,1.5) {\scriptsize $2$};
	\draw[<-] (j21) edge (j22) (j22) edge (j23);	
	\draw[<-, gray!40] (j21) edge (j24);	
	
	\node [mystyle] (j31) at (7,1.5) {\scriptsize $b_0$};
	\node [mystyle] (j32) at (7.5,1) {\scriptsize $1$};
	\node [mystyle] (j33) at (8,1) {\scriptsize $3$};
	\node [mystyle] (j34) at (8.5,1) {\scriptsize $5$};
	\node [mystyleLight, fill=olive!20] (j35) at (7.5,1.5) {\scriptsize $2$};
	\node [mystyleLight, fill=olive!20] (j36) at (8,1.5) {\scriptsize $4$};
	\draw[<-] (j31) edge (j32) (j32) edge (j33) (j33) edge (j34) ;	
	\draw[<-, gray!40] (j31) edge (j35) (j35) edge (j36);

	\node[]() at (0,0.5){$j$}; \draw[->] (.1,0.5) -- (10,0.5); \node[]() at (10,.8){$t$};	
	
	\draw[draw=black] (0,3.2) rectangle (2,3.8); \node[]() at (1,4){{\sf read}$()$, l=3};
	
%
	
	\draw[dash pattern={on 7pt off 7pt}] (5,3.7) -- (2.2,3.7) -- (2.2, .2) -- (3.3, 0.2) -- (3.3, 1.8) --(5,1.8) -- (5,3.7); 
	
	\draw[dash pattern={on 3pt off 3pt}] (9,3.7) -- (5.4,3.7) -- (5.4, 1.7) -- (3.5, 1.7) -- (3.5,0.3) -- (9, 0.3); 
	
	\draw[dash pattern={on 7pt off 7pt}] (10.5, 3.3) rectangle (11,3.8);
	\node[]() at (13.5,3.5){\scriptsize set with blockchains mcps $l' < l$};
	\draw[dash pattern={on 3pt off 3pt}] (10.5, 2.3) rectangle (11,2.8);
	\node[]() at (13.5,2.5){\scriptsize set with blockchains mcps $l'\geq l$};	
	\end{tikzpicture}
	\caption{Sets for the Eventual Prefix Property.}\label{fig:EventualHistoryPrefix}
\end{subfigure}
	\caption{Concurrent history that satisfies the Eventual BT consistency criterion. In such scenario $f$ selects the longest blockchain and the blockchain score is the length $l$. In case (a) and case (b) the concurrent history is the same but different sets are outlined.}\label{fig:EventualHistory}
\end{figure}

Figure \ref{fig:EventualHistory} shows a concurrent history that satisfies the Eventual prefix property but not the Strong prefix one. Strong Prefix is not satisfied as blockchain\footnote{For ease of readability we extend the notation $b_i^ \frown b_j$ to represent concatenated blocks in a blockchain.} $b_0^\frown 1$ returned  from the first read() at process $j$ is not a prefix of  blockchain $b_0 ^\frown 2 ^\frown 4$ returned from the first read at process $i$. Note that we adopt the same conventions as for the example depicted in Figure~\ref{fig:StrongHistory} regarding the score, length and append() operations. 
We assume that the Block validity property is satisfied. The Local monotonic read property is easily verifiable. In both Figures~\ref{fig:EventualHistoryEGT} and~\ref{fig:EventualHistoryPrefix}, the first {\sf read}$()$ $r$ operation at $i$, enclosed in a black rectangle, is taken as reference to check the consistency criterion (the criterion has to be iteratively verified for each {\sf read}$()$ operation). Let $l$ be the score of the blockchain returned by $r$. In Figure \ref{fig:EventualHistoryPrefix} we can identify two sets, enclosed in rectangles defined by different patterns: \emph{(i)}  the finite set of {\sf read}$()$ operations sharing a maximum common prefix score (mcps) smaller than $l$ (the set to check for the satisfiability of the Eventual Prefix property), and \emph{(ii)}  the infinite set of {\sf read}$()$ operations such that for each couple of them $bc$, $bc'$, {\sf mcps}$(bc,bc') \geq l$. We can iterate the same reasoning for each {\sf read}$()$ operation in $H$. Thus $H$ satisfies the Eventual Prefix property.
Figure \ref{fig:UncocherentHistory}  shows a history that does not satisfy any consistency criteria defined so far.

\begin{figure}[h]
	\begin{subfigure}{.5\textwidth}
	\begin{tikzpicture}
	\node [mystyle] (i11) at (0.5,3.5) {\scriptsize $b_0$};
	\node [mystyle, gray!40] (i14) at (1,3) {\scriptsize $1$};
	\node [mystyle, fill=olive] (i12) at (1,3.5) {\scriptsize $2$};
	\node [mystyle, fill=olive] (i13) at (1.5,3.5) {\scriptsize $4$};
	\draw[<-] (i11) edge (i12) (i12) edge (i13) ;
	\draw[<-, gray!40] (i11)edge(i14);

	\node [mystyle] (i21) at (3,3.5) {\scriptsize $b_0$};
	\node [mystyle, fill=olive] (i22) at (3.5,3.5) {\scriptsize $2$};
	\node [mystyle, fill=olive] (i23) at (4,3.5) {\scriptsize $4$};
	\node [mystyle,gray!40] (i24) at (3.5,3) {\scriptsize $1$};
	\node [mystyle,gray!40] (i25) at (4,3) {\scriptsize $3$};
	\draw[<-] (i21) edge (i22) (i22) edge (i23);	
	\draw[<-, gray!40] (i21) edge  (i24) (i24)edge(i25);	
	
	\node [mystyle] (i31) at (6,3.5) {\scriptsize $b_0$};
	\node [mystyle, fill=olive] (i35) at (6.5,3.5) {\scriptsize $2$};
	\node [mystyle, fill=olive] (i36) at (7,3.5) {\scriptsize $4$};
	\node [mystyle, fill=olive] (i37) at (7.5,3.5) {\scriptsize $6$};
	\node [mystyleLight] (i32) at (6.5,3) {\scriptsize $1$};
	\node [mystyleLight] (i33) at (7,3) {\scriptsize $3$};
	\draw[<-, gray!40] (i31) edge (i32)  (i32) edge (i33);
	\draw[<-]  (i31) edge  (i35) (i35) edge(i36) (i36) edge (i37);
	
	\node[]() at (0,2.5){$i$}; \draw[->] (.1,2.5) -- (10,2.5);	\node[]() at (10,2.8){$t$};
	
	\node [mystyle] (j11) at (2.5,1.5) {\scriptsize $b_0$};
	\node [mystyle] (j12) at (3,1) {\scriptsize $1$};
	\draw[<-] (j11) edge (j12) 
	;
	
	\node [mystyle] (j21) at (4,1.5) {\scriptsize $b_0$};
	\node [mystyle] (j22) at (4.5,1) {\scriptsize $1$};
	\node [mystyle] (j23) at (5,1) {\scriptsize $3$};
	\node [mystyleLight, fill=olive!20] (j24) at (4.5,1.5) {\scriptsize $2$};
	\draw[<-] (j21) edge (j22) (j22) edge (j23);	
	\draw[<-, gray!40] (j21) edge (j24);	
	
	\node [mystyle] (j31) at (7,1.5) {\scriptsize $b_0$};
	\node [mystyle] (j32) at (7.5,1) {\scriptsize $1$};
	\node [mystyle] (j33) at (8,1) {\scriptsize $3$};
	\node [mystyle] (j34) at (8.5,1) {\scriptsize $5$};
	\node [mystyleLight, fill=olive!20] (j35) at (7.5,1.5) {\scriptsize $2$};
	\node [mystyleLight, fill=olive!20] (j36) at (8,1.5) {\scriptsize $4$};
	\draw[<-] (j31) edge (j32) (j32) edge (j33) (j33) edge (j34) ;	
	\draw[<-, gray!40] (j31) edge (j35) (j35) edge (j36);

	\node[]() at (0,0.5){$j$}; \draw[->] (.1,0.5) -- (10,0.5); \node[]() at (10,.8){$t$};	
	
	\draw[] (0,3.2) rectangle (2,3.8); \node[]() at (1,4){{\sf read}$()$, l=3};
	
	\draw[dash pattern={on 7pt off 2pt on 1pt off 3pt}] (9,3.8) -- (5.7,3.8) -- (5.7,2) -- (6.5,2) -- (6.5, 0.7) -- (9, 0.7) ; 
	
	\draw[dotted] (2.1, 0.7) -- (2.1,3.9) --(5.5,3.9) -- (5.5, 0.7) -- (2.1, 0.7);
	
	
	
\draw[dotted] (10.5, 2.3) rectangle (11,2.8);
\node[]() at (13.5,2.5){\scriptsize set with each bc score $l' \leq l$};	
\draw[dash pattern={on 7pt off 2pt on 1pt off 3pt}] (10.5, 3.3) rectangle (11,3.8);
\node[]() at (13.5,3.5){\scriptsize set with each bc score $l'>l$};

	\end{tikzpicture}
		\caption{Sets for the Ever Growing Tree property.}
\end{subfigure}
\\
\begin{subfigure}[ASD]{.5\textwidth}
		\begin{tikzpicture}
	\node [mystyle] (i11) at (0.5,3.5) {\scriptsize $b_0$};
	\node [mystyle, gray!40] (i14) at (1,3) {\scriptsize $1$};
	\node [mystyle, fill=olive] (i12) at (1,3.5) {\scriptsize $2$};
	\node [mystyle, fill=olive] (i13) at (1.5,3.5) {\scriptsize $4$};
	\draw[<-] (i11) edge (i12) (i12) edge (i13) ;
	\draw[<-, gray!40] (i11)edge(i14);

	\node [mystyle] (i21) at (3,3.5) {\scriptsize $b_0$};
	\node [mystyle, fill=olive] (i22) at (3.5,3.5) {\scriptsize $2$};
	\node [mystyle, fill=olive] (i23) at (4,3.5) {\scriptsize $4$};
	\node [mystyle,gray!40] (i24) at (3.5,3) {\scriptsize $1$};
	\node [mystyle,gray!40] (i25) at (4,3) {\scriptsize $3$};
	\draw[<-] (i21) edge (i22) (i22) edge (i23);	
	\draw[<-, gray!40] (i21) edge  (i24) (i24)edge(i25);	
	
	\node [mystyle] (i31) at (6,3.5) {\scriptsize $b_0$};
	\node [mystyle, fill=olive] (i35) at (6.5,3.5) {\scriptsize $2$};
	\node [mystyle, fill=olive] (i36) at (7,3.5) {\scriptsize $4$};
	\node [mystyle, fill=olive] (i37) at (7.5,3.5) {\scriptsize $6$};
	\node [mystyleLight] (i32) at (6.5,3) {\scriptsize $1$};
	\node [mystyleLight] (i33) at (7,3) {\scriptsize $3$};
	\draw[<-, gray!40] (i31) edge (i32)  (i32) edge (i33);
	\draw[<-]  (i31) edge  (i35) (i35) edge(i36) (i36) edge (i37);
	
	\node[]() at (0,2.5){$i$}; \draw[->] (.1,2.5) -- (10,2.5);	\node[]() at (10,2.8){$t$};
	
	\node [mystyle] (j11) at (2.5,1.5) {\scriptsize $b_0$};
	\node [mystyle] (j12) at (3,1) {\scriptsize $1$};
	\draw[<-] (j11) edge (j12) 
	;
	
	\node [mystyle] (j21) at (4,1.5) {\scriptsize $b_0$};
	\node [mystyle] (j22) at (4.5,1) {\scriptsize $1$};
	\node [mystyle] (j23) at (5,1) {\scriptsize $3$};
	\node [mystyleLight, fill=olive!20] (j24) at (4.5,1.5) {\scriptsize $2$};
	\draw[<-] (j21) edge (j22) (j22) edge (j23);	
	\draw[<-, gray!40] (j21) edge (j24);	
	
	\node [mystyle] (j31) at (7,1.5) {\scriptsize $b_0$};
	\node [mystyle] (j32) at (7.5,1) {\scriptsize $1$};
	\node [mystyle] (j33) at (8,1) {\scriptsize $3$};
	\node [mystyle] (j34) at (8.5,1) {\scriptsize $5$};
	\node [mystyleLight, fill=olive!20] (j35) at (7.5,1.5) {\scriptsize $2$};
	\node [mystyleLight, fill=olive!20] (j36) at (8,1.5) {\scriptsize $4$};
	\draw[<-] (j31) edge (j32) (j32) edge (j33) (j33) edge (j34) ;	
	\draw[<-, gray!40] (j31) edge (j35) (j35) edge (j36);

	\node[]() at (0,0.5){$j$}; \draw[->] (.1,0.5) -- (10,0.5); \node[]() at (10,.8){$t$};	
	
	\draw[draw=black] (0,3.2) rectangle (2,3.8); \node[]() at (1,4){{\sf read}$()$, l=3};
	
	
	
	\draw[dash pattern={on 7pt off 7pt}] (9,3.7) -- (2.2,3.7) -- (2.2, .2) --  (9,.2); 
	
	
%
	\draw[dash pattern={on 7pt off 7pt}] (10.5, 3.3) rectangle (11,3.8);
	\node[]() at (13.5,3.5){\scriptsize set with blockchains mcps $l' < l$};

	\end{tikzpicture}
		\caption{Sets for the Eventual Prefix Property.}
\end{subfigure}
	\caption{Concurrent history that does not satisfy any BT consistency criteria. In such scenario $f$ selects the longest blockchain and the blockchain score is the length $l$.}\label{fig:UncocherentHistory}
\end{figure}
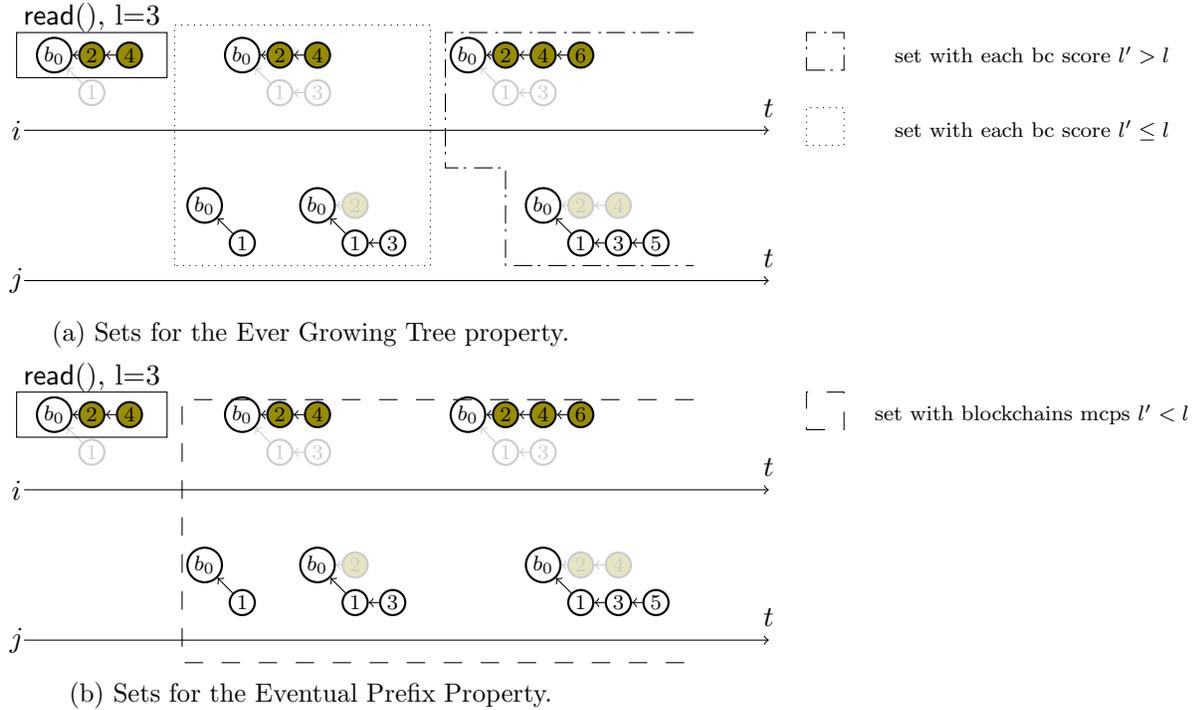

\paragraph{Relationships between $EC$ and $SC$.}

Let us denote by 
$\mathcal{H}_{EC} $ and by $\mathcal{H}_{SC} $ the set of histories satisfying respectively the $EC$ and the $SC$ consistency criteria.  

\begin{theorem}\label{th:strongImpliesEventual}
Any history $H$ satisfying $SC$ criterion satisfies $EC$  and $ \exists H$ satisfying  $EC$ that does not satisfy $SC$, i.e.,  
$\mathcal{H}_{SC} \subset \mathcal{H}_{EC}$. 

\end{theorem}

\begin{proof}
$E C \leq SC$ implies that $\mathcal{H}_{SC} \subset \mathcal{H}_{E C} $, and $\mathcal{H}_{SC} \subset \mathcal{H}_{E C} $ implies that $\forall H \in \mathcal{H}_{SC} \Rightarrow H\in \mathcal{H}_{E C}$. By hypothesis, $H$  verifies the Ever Growing Tree property, thus $\forall e_{rsp}(r) \in E(a^*,r^*)$ with $s=score(e_{rsp}(r):bc)$ then set $\{ e_{inv}(r') \in  E |e_{rsp}(r) \nearrow e_{inv}(r'), score(e_{rsp}(r''):bc)\leq s \}$ is finite, and thus, there is an infinite set  $\{ e_{inv}(r') \in  E |e_{rsp}(r) \nearrow e_{inv}(r'), score(e_{rsp}(r''):bc)> s \}$. The  Strong prefix property guarantees that $\forall e_{rsp}(r), e_{rsp}(r') \in H,  (e_{rsp}(r):bc \sqsubseteq e_{rsp}(r):bc') \vee (e_{rsp}(r):bc \sqsubseteq e_{rsp}(r'):bc')$, thus in this infinite set,  all the ${\sf read}()$ operations return blockchains sharing the same maximum prefix whose score is at least $s+1$, which satisfies the Eventual prefix property. The Eventual Prefix property demands that for each  $\forall e_{rsp}(r) \in E(a,r^*)$ with $s=score(e_{rsp}(r):bc)$ there is an infinite set defined as $\{ (e_{rsp}(r_h), e_{rsp}(r_k)) \in E_{r}^2 | h\neq k, {\sf mpcs}(e_{rsp}(r_h):bc_h, e_{rsp}(r_k):bc_k) \geq s \}$ where \(E_{r}\) denotes the set of response events of read operations that occurred after \(r\) response.
To conclude the proof we need to find a $H\in\mathcal{H}_{E C}$ and $H\not \in \mathcal{H}_{SC}$. 
Any $H$ in which at least two  {\sf read}$()$ operations return a blockchain sharing the same prefix but diverging in their suffix violate the Strong prefix property, which concludes the proof.
\end{proof}

Let us remark that the BlockTree allows at any time to create a new branch in the tree, which is called a \textit{fork} in the blockchain literature. Moreover, an append is successful only if the input block is valid with respect to a predicate. This means that histories with no append operations are trivially admitted. 
In the following we will introduce a new abstract data type called Token Oracle that when combined with the BlockTree will help in \emph{(i)} validating blocks and \emph{(ii)} controlling forks. We will first formally introduce the Token Oracle in Section \ref{sec:oracle} and then we will define the properties on the BlockTree augmented with the Token Oracle in Section \ref{sec:composition}.


\subsection{Token oracle $\Theta$-ADT}\label{sec:oracle}
In this section we formalize the Token Oracle $\Theta$ to capture the creation of blocks in the BlockTree structure. The block creation process requires that the new block must be closely related to an already existing valid block in the BlockTree structure. We abstract this implementation-dependent process by assuming that a process will obtain the right to chain a new block $b_{\ell}$ to $b_{h}$ if it successfully gains a token  $tkn_{h}$  from the token oracle $\Theta$. Once obtained, the proposed block $b_{\ell}$ is considered as valid, and will be denoted by $b_{\ell}^{tkn_{h}}$. By construction  $b_{\ell}^{tkn_{h}}\in \mathcal{B'}$. 
In the following, in order to be as much general as possible, we model blocks as objects.
More formally, when a process wants to access a generic object $obj_h$, it invokes the {\sf getToken}$(obj_h,obj_\ell)$ operation with object $obj_\ell$ from set $\mathcal{O}=\{obj_1,obj_2, \dots\}$.  If {\sf getToken}$(obj_h,obj_\ell)$ operation is successful, it returns an object $ obj_{\ell}^{tkn_{h}} \in \mathcal{O'}$, where \emph{(i)} ${tkn_{h}}$ is the token required to access object $obj_h$ and \emph{(ii)} each object $obj_k \in \mathcal{O'}$ is valid with respect to predicate $P$, i.e. $P(obj_k)=\top$.  We say that a token is {\em generated} each time it is provided to a process and it is {\em consumed} when the oracle grants the right to connect it to the previous object. Each token can be consumed at most once. To consume a token we define the token consumption {\sf consumeToken}$( obj_{\ell}^{tkn_{h}})$ operation, where the consumed token $tkn_{h}$ is the token required for the object $obj_h$.  A maximal number of tokens $k$ for an object $obj_h$ is managed by the oracle. 
The {\sf consumeToken}($obj_{\ell}^{tkn_{h}})$ side-effect on the state is the insertion of the object $obj_{\ell}^{tkn_{h}}$ in a set $K_h$ as long as the cardinality of such set is less than $k$.

In the following we specify two token oracles, which differ in the way tokens are managed. 
The first oracle, called  {\em prodigal} and denoted by $\Theta_P$, has no upper bound on the number of tokens consumed for an object, while the second oracle $\Theta_F$,  called {\em frugal}, and denoted by $\Theta_F$, assures controls that no more than $k$ token can be consumed for each object. 

$\Theta_P$ when combined with the BlockTree abstract data type will only help in validating blocks, while $\Theta_F$ manages tokens in a more controlled way to guarantee that no more than $k$ forks can occur on a given block. 

\subsubsection{$\Theta_P$-ADT and $\Theta_F$-ADT definitions}
\label{sec:theta_p-theta_f-adt}


For both oracles, when {\sf getToken}$(obj_k,obj_h)$ operation is invoked, the oracle provides a token with a certain probability $p_{\alpha_i}>0$ where $\alpha_i$ is a \say{merit} parameter characterizing  the invoking process $i$.~\footnote{The merit parameter can reflect for instance the hashing power of the invoking process.} 
Note that the oracle knows $\alpha_i$ of the invoking process $i$, which might be unknown to the process itself. For each merit $\alpha_i$, the state of the token oracle embeds an infinite tape where each cell of the tape contains either $tkn$ or $\bot$. Since each tape is identified by a specific $\alpha_i$ and $p_{\alpha_i}$,  we assume that each tape contains a pseudorandom sequence of values in $\{tkn,\bot\}$ depending on $\alpha_i$.~\footnote{We assume a pseudorandom sequence mostly indistinguishable from a Bernoulli sequence  consisting of a finite or infinite number of independent random variables $X1, X2 , X3, \dots $ such that \emph{(i)}  for each $k$, the value of $X_k$ is either $tkn$ or $\bot$; and \emph{(ii)} $\forall X_k$ the probability that $X_k=tkn$ is $p_{\alpha_i}$.} When a {\sf getToken}$(obj_k,obj_h)$ operation is invoked by a process with merit $\alpha_i$, the oracle pops the first cell from the tape associated to $\alpha_i$, and a token is provided to the process if that cell contains $tkn$.

Both oracles also enjoy  an infinite array of sets, one for each object, which is populated each time a token is consumed for a specific object. When the set cardinality reaches $k$ then no more tokens can be consumed for that object. For a sake of generality, $\Theta_P$ is defined as $\Theta_F$ with $k=\infty$ while for $\Theta_F$ a predetermined $k \in \mathbb{N}$ is specified.

\begin{figure}[h]
 	\centering
\begin{tikzpicture}
	\node[]() at (-9.8,2.5){
	\begin{tabular}{|c|c|c|c|c}
	\hline
	$\{\}_1$&$\{\}_2$&$\{\}_3$&$\{\}_4$&$\dots$\\
	\hline
	\end{tabular}
	};
	\node[]() at (-12.7,2.5) {$K$};
	\node[]() at (-11.6,2) {\scriptsize $obj_1$};
	\node[]() at (-10.8,2) {\scriptsize $obj_2$};
	\node[]() at (-9.8,2) {\scriptsize $obj_3$};
	\node[]() at (-8.9,2) {\scriptsize $obj_4$};
	\node[]() at (-7.9,2) {\scriptsize $\dots$};

%
	

	\node[]() at (-2,2){
	\begin{tabular}{|c|c|c|c|c|c|c|c}
	\hline
	{\scriptsize $tkn$}&$\bot$&$\bot$&{\scriptsize $tkn$}&$\bot$&$\bot$&$\bot$&$\dots$\\
	\hline
	\end{tabular}
	};
	\node[]() at (-6.2,2) {$tape_{\alpha_2}$};

	\node[]() at (-2,3){
	\begin{tabular}{|c|c|c|c|c|c|c|c}
	\hline
	$\bot$&$\bot$&$\bot$&$\bot$&$\bot$&$\bot$&{\scriptsize $tkn$}&$\dots$\\
	\hline
	\end{tabular}
	};
	\node[]() at (-6.2,3) {$tape_{\alpha_1}$};
	\node[]() at (-6,1) {$\vdots$};	
\end{tikzpicture}
\caption{The $\Theta_F$ abstract state. The infinite $K$ array, where at the beginning each set is initialized as empty and the infinite set of infinite tapes, one for each merit $\alpha_i$ in $\mathcal{A}$.}\label{fig:oracleState}
 \end{figure}
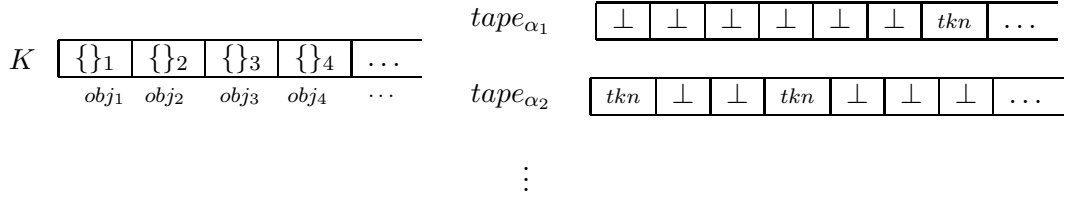
 
We  first introduce some definitions and notations.
\begin{itemize}
	\item $\mathcal{O}=\{obj_1,obj_2, \dots\}$, infinite set of generic objects uniquely identified by their index $i$;
	\item $\mathcal{O'} \subset \mathcal{O}$,  the subset of objects  valid with respect to predicate $P$, i.e. $ \forall obj'_i \in  \mathcal{O'}, P(obj'_i)=\top$.
	\item $\mathfrak{T}=\{tkn_1,tkn_2, \dots\}$  infinite set of tokens;
	\item $\mathcal{A}=\{\alpha_1, \alpha_2, \dots\}$ an infinite set of rational values;
	\item $\mathcal{M}$ is a countable not empty set of mapping functions $ m(\alpha_i)$ that generate an infinite pseudo random tape $tape_{\alpha_i}$ such that the probability to have in a cell the string $tkn$ is related to a specific $\alpha_i$, $m \in \mathcal{M}: \mathcal{A} \rightarrow \{tkn ,\bot\}^*$; 
	\item ${K}[\ ]$ is a infinite array of sets (one per object) of elements in $\mathcal{O'}$. All the sets are initialized as empty and can be fulfilled with at most $k$ elements, where $k\in \mathbb{N}$ is a parameter of the oracle ADT;
	\item ${pop}: \{tkn ,\bot\}^* \rightarrow \{tkn ,\bot\}^*$, ${pop}(a\cdot w)=w$;
	\item ${head}: \{tkn ,\bot\}^* \rightarrow \{tkn ,\bot\}^*$, ${head}(a\cdot w)=a$;	

	\item $add : \{K\} \times \mathbb{N} \times \mathcal{O'} \rightarrow \{K\}$, $add(K, i, obj_{\ell}^{tkn_{h}} )=K$ : $K[i]=K[i] \cup \{obj_{\ell}^{tkn_{h}}\}$ {\bf if} $|K[i]|<k$; {\bf else} $K[i]=K[i]$;
	\item $get : \{K\} \times  \mathbb{N} \rightarrow \mathbb{N}$, $get(K, i)=K[i]$;   
\end{itemize}

\definition{\textbf{($\Theta_F$-ADT Definition)}.}\label{def:btadt}
The $\Theta_F$ Abstract Data type is the $6$-tuple 
$\Theta_F$-ADT =$\langle$ A= \{{\sf getToken}$(obj_h,obj_\ell)$, {\sf consumeToken}$( obj_{\ell}^{tkn_{h}}):
obj_h, obj_{\ell}^{tkn_{h}} \in \mathcal{O'}, obj_\ell \in \mathcal{O}, tkn_h \in \mathfrak{T}\}$, B= $\mathcal{O'}  \cup Boolean$, Z= $m(\mathcal{A})^* \times \{K\} \times k \cup \{pop, head, dec, get\}$, $\xi_0, \tau, \delta \rangle$, where the transition function  $\tau:Z \times A \rightarrow Z$ is defined by 
\begin{itemize}
	\item  $\tau( (\{tape_{\alpha_1},\dots, tape_{\alpha_i}, \dots\}, K, k) ,{\sf getToken}(obj_h,obj_\ell))=(\{tape_{\alpha_1},\dots, pop(tape_{\alpha_i}), \dots\}, K, k)$ with $\alpha_i$ the merit of the invoking process;
	\item  $\tau( (\{tape_{\alpha_1},\dots, tape_{\alpha_i}, \dots\}, K, k), {\sf consumeToken}(obj_{\ell}^{tkn_{h}}))= (\{tape_{\alpha_1},\dots, tape_{\alpha_i}, \dots\},$ $ {add}(K, h,obj_{\ell}^{tkn_{h}}))$, if
	$tkn_{h}\in \mathfrak{T}$  ;  $\{ (\{tape_{\alpha_1},\dots, tape_{\alpha_i}, \dots\}, K, k)\}$ otherwise.
\end{itemize}
and the output function $\delta:Z \times A \rightarrow B$ is defined by 
\begin{itemize}
	\item  $\delta((\{tape_{\alpha_1},\dots, tape_{\alpha_i}, \dots\}, K,k) ,{\sf getToken}(obj_h,obj_\ell))= obj_{\ell}^{tkn_{h}}:  obj_{\ell}^{tkn_{h}} \in \mathcal{O'},tkn_h \in \mathfrak{T}$, if $head(tape_{\alpha_i})=tkn$ with $\alpha_i$ the merit of the invoking process;  $\bot$ otherwise;
	\item  $\delta( (\{tape_{\alpha_1},\dots, tape_{\alpha_i}, \dots\}, K,k) , {\sf consumeToken}(obj_{\ell}^{tkn_{h}}))= {get}(K, h)$.
	
\end{itemize}

\definition{\textbf{($\Theta_P$-ADT Definition)}.}
The $\Theta_P$ Abstract Data type is defined as the $\Theta_F$-ADT with $k=\infty$.\\

\noindent Figure \ref{fig:Theta-ADT} shows a possible path of the transition system defined  by the $\Theta_F$ and  $\Theta_P$-ADTs.

\begin{figure}
	\begin{tikzpicture}[->, auto]
	\tikzset{stato/.style={draw,circle, node distance=4.9cm}}
	\tikzset{dettaglio/.style={black!70}}
	\tikzset{blocco/.style={
			circle,
			inner sep=1pt,
			thick,
			align=center,
			draw=gray,
		}
	}
	
	\node[stato] (0) at (0,0) {$\xi_0 $};
	
	\node[dettaglio, below of=0] (s0) {$\xi_0 = \{$
		\begin{tikzpicture}[scale=0.5, every node/.style={scale=0.5}]	
		\node[] () at (0,7){};
		\node[]() at (-1,4){
			\begin{tabular}{|c|c|c|c}
			\hline
			$\{\}$&$\{\}$&$\{\}$&$\dots$\\
			\hline
			\end{tabular}
		};
		\node[]() at (-4,4) {$K$};
		
		\node[]() at (-1,3){
			\begin{tabular}{|c|c|c|c}
			\hline
			{\scriptsize $tkn$}&$\bot$&$\bot$&$\dots$\\
			\hline
			\end{tabular}
		};
		\node[]() at (-4,2) {$tape_{\alpha_2}$};
		
		\node[]() at (-1,2){
			\begin{tabular}{|c|c|c|c}
			\hline
			$\bot$&$\bot$&{\scriptsize $tkn$}&$\dots$\\
			\hline
			\end{tabular}
		};
		\node[]() at (-4,3) {$tape_{\alpha_1}$};
		\node[]() at (-4,1.5) {$\vdots$};	
		\end{tikzpicture}
		$, k\}$};

	\node[stato, right of=0] (1) {$\xi_1$}; 
	\node[dettaglio, below of=1] (s1) {$\xi_1= \{$			
		\begin{tikzpicture}[scale=0.5, every node/.style={scale=0.5}]	
		\node[] () at (0,7){};
		\node[]() at (-1,4){
			\begin{tabular}{|c|c|c|c}
			\hline
			$\{\}$&$\{\}$&$\{\}$&$\dots$\\
			\hline
			\end{tabular}
		};
		\node[]() at (-4,4) {$K$};
		
		\node[]() at (-1,3){
			\begin{tabular}{|c|c|c|c}
			\hline
			{$\bot$}&$\bot$&$\bot$&$\dots$\\
			\hline
			\end{tabular}
		};
		\node[]() at (-4,2) {$tape_{\alpha_2}$};
		
		\node[]() at (-1,2){
			\begin{tabular}{|c|c|c|c}
			\hline
			$\bot$&$\bot$&{\scriptsize $tkn$}&$\dots$\\
			\hline
			\end{tabular}
		};
		\node[]() at (-4,3) {$tape_{\alpha_1}$};
		\node[]() at (-4,1.5) {$\vdots$};	
		\end{tikzpicture}
		$,k\}$};
	
	\path (0) edge [bend left] node[align=center] {\scriptsize {\sf getToken}$(obj_1, obj_k)/ {obj_k^{tkn_1}}$ \\ {\scriptsize{\bf if }$ pop(tape_{\alpha_1})={tkn}$} } (1);
	
	
	\node[stato, right of=1] (2) {$\xi_2$}; 
	\node[dettaglio, below of=2] (21) {$\xi_2= \{$
		\begin{tikzpicture}[scale=0.5, every node/.style={scale=0.5}]	
		\node[] () at (0,7){};
		\node[]() at (-1,4){
			\begin{tabular}{|c|c|c|c}
			\hline
			$\{obj_k^{tkn_1}\}$&$\{\}$&$\{\}$&$\dots$\\
			\hline
			\end{tabular}
		};
		\node[]() at (-4,4) {$K$};
		
		\node[]() at (-1.5,3){
			\begin{tabular}{|c|c|c|c}
			\hline
			{$\bot$}&$\bot$&$\bot$&$\dots$\\
			\hline
			\end{tabular}
		};
		\node[]() at (-4,2) {$tape_{\alpha_2}$};
		
		\node[]() at (-1.5,2){
			\begin{tabular}{|c|c|c|c}
			\hline
			$\bot$&$\bot$&{\scriptsize $tkn$}&$\dots$\\
			\hline
			\end{tabular}
		};
		\node[]() at (-4,3) {$tape_{\alpha_1}$};
		\node[]() at (-4,1.5) {$\vdots$};	
		\end{tikzpicture}
		$,k \}$};	
	
	\path (1) edge [bend left] node[align=center] {\scriptsize {\sf consumeToken}$({obj_k^{tkn_1}})/ \{obj_k^{tkn_1}\}$ \\ \scriptsize {\bf if} $|K[1]|<k \wedge b_k^{tkn_1} \in \mathfrak{T}$} (2);
\end{tikzpicture}
\caption{A possible path of the transition system defined  by the $\Theta_F$ and  $\Theta_P$-ADTs. We use the following syntax on the edges: {\sf operation}/{\sf output}.}\label{fig:Theta-ADT}
\end{figure}
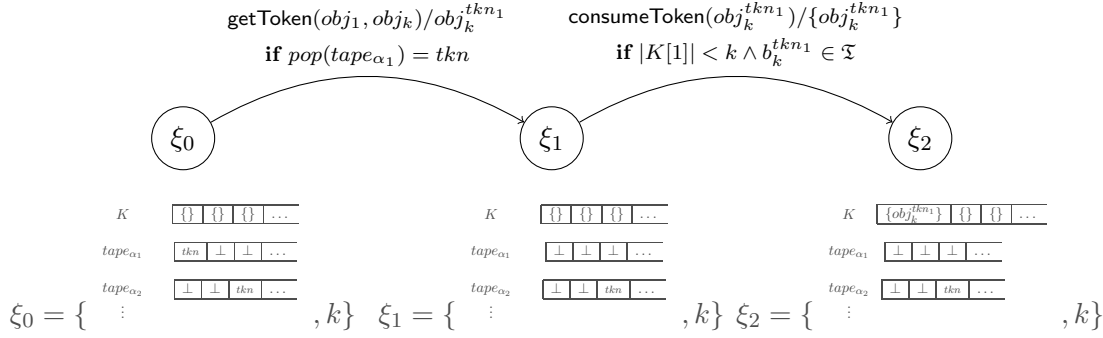

\subsection{BT-ADT augmented with $\Theta$ Oracles}
In this section we augment the BT-ADT  with $\Theta$ oracles and we analyze the histories generated by their combination.
Specifically, we define a refinement of the {\sf append}$(b_\ell)$ operation of the BT-ADT with the oracle operations which triggers the {\sf getToken}($b_h\leftarrow$last\_block$(f(bt)), b_\ell$) operation as long as it returns a token on $b_k$, i.e., ${b_\ell}^{tkn_h}$ which is a valid block in $\mathcal{B'}$. Once obtained, the token is consumed and the append terminates, i.e. the block ${b_\ell}^{tkn_h}$ is appended to the block $h$ in the blockchain $f(bt)$ ($\{b_0\}^\frown {f(bt)}|_h^ {\frown} \{b_\ell\}$). Notice that those two operations and the concatenation occur atomically.\\ 
We say that the $BT$-$ADT$ augmented with $\Theta_F$ or $\Theta_P$ oracle is a \textit{refinement}  $\mathfrak{R}(BT$-$ADT, \Theta_F)$ or $\mathfrak{R}(BT$-$ADT, \Theta_P)$ respectively.

Let us define the following auxiliary function:
\begin{itemize}
	\item $evaluate$: $\mathcal{B} \times B^{\Theta} \rightarrow bool$. $evaluate(b, \delta_b \circ \delta_a^*$ )= {\sf true} if $(\exists h: b^{tkn_h}\in  \delta_b \wedge (\exists X : b^{tkn_h} \in X \wedge X \in \delta_a^*))$; {\sf false} otherwise.
\end{itemize}

\begin{definition}\label{def:refinement}[$\mathfrak{R}(BT$-$ADT, \Theta_F)$  refinement]
Given the   BT-ADT=$\langle A, B, Z,\xi_0, \tau, \delta \rangle$, 
and the 
$\Theta_F$-ADT =($A^{\Theta}, B^{\Theta}, Z^{\Theta}$, $\xi_{0}^{\Theta}, \tau^{\Theta}, \delta^{\Theta}$), 
we have 
$\mathfrak{R}(BT-ADT, \Theta_F)$=$\langle A'= A \cup A^{\Theta}, B'=B \cup B^{\Theta}, Z'=Z \cup Z^{\Theta}$, $\xi_{0}'=\xi_0  \cup \xi_{0}^{\Theta},\tau', \delta' \rangle,$ where the transition function $\tau':Z' \times A' \rightarrow Z'$ is defined by
\begin{itemize}
	\item  $\tau_a=\tau^\prime (( \{tape_{\alpha_1},\dots, tape_{\alpha_i}, \dots\}, K,k, bt,f, P), {\sf getToken}(b_k \leftarrow \textnormal{last\_block}(bt),b_\ell))=$\\ $ (\{tape_{\alpha_1},\dots, pop(tape_{\alpha_i}), \dots\}, K,k, bt,f,P)$;
	\item  $\tau_b=\tau^\prime (( \{tape_{\alpha_1},\dots, tape_{\alpha_i}, \dots\}, K,k, bt,f,P), {\sf consumeToken}(b_{\ell}^{tkn_{h}}))=$\\ $ (\{tape_{\alpha_1},\dots, tape_{\alpha_i}, \dots\}, {add}(K, h, b_{\ell}^{tkn_{h}}),k, \{b_0\}^\frown {f(bt)}|_h^ \frown \{b_\ell\},f,P)$ if
	$tkn_{h}\in \mathfrak{T} \wedge b_{\ell}^{tkn_{h}} \in get(K,l)$ ;  $( \{tape_{\alpha_1},\dots, tape_{\alpha_i}, \dots\}, K,k, bt,f,P)$ otherwise;
	 \item $\tau'( (\{tape_{\alpha_1},\dots, tape_{\alpha_i}, \dots\}, K,k, bt,f,P), {\sf append}(b))=\tau_b \circ \tau_a^*$
	 
	 where $ \tau_b \circ \tau_a^*$ is the repeated application of $\tau_a$ until \\  $\delta_a((\{tape_{\alpha_1},\dots, tape_{\alpha_i}, \dots\}, K,k, bt,f,P), {\sf getToken}(b_k \leftarrow \textnormal{last\_block}(bt),b_\ell) )= b_{\ell}^{tkn_{h}}$ concatenated with the $\tau_b$ application;
	 \item  $\tau'(\{tape_{\alpha_1},\dots, tape_{\alpha_i}, \dots\}, K,k, bt,f,P), {\sf read}()= bt$.
\end{itemize}
and the output function $\delta':Z' \times A' \rightarrow B'$ is defined by:
\begin{itemize}
	\item  $\delta_a=\delta^\prime((\{tape_{\alpha_1},\dots, tape_{\alpha_i}, \dots\}, K,k, bt,f,P) , {\sf getToken}(b_k \leftarrow \textnormal{last\_block}(bt),b_\ell) )= b_{\ell}^{tkn_{h}}:  b_{\ell}^{tkn_{h}} \in \mathcal{B'},tkn_h \in \mathfrak{T}$, if $head(tape_{\alpha_i})=tkn$ with $\alpha_i$ the merit of the invoking process;  $\bot$ otherwise;
	\item  $\delta_b=\delta^\prime ( (\{tape_{\alpha_1},\dots, tape_{\alpha_i}, \dots\}, K,k,  bt,f,P) , {\sf consumeToken}(obj_{\ell}^{tkn_{h}}))= {get}(K, h)$; 
	 \item $\delta'( (\{tape_{\alpha_1},\dots, tape_{\alpha_i}, \dots\}, K,k, bt,f,P), {\sf append}(b))=evaluate(b,\delta_b \circ \delta_a^*$), where $ \delta_b \circ \delta_a^*$ is the repeated application of $\delta_a$ until  $\delta_a((\{tape_{\alpha_1},\dots, tape_{\alpha_i}, \dots\}, K,k, bt,f,P) ,$ $ {\sf getToken}(\textnormal{last\_block}(bt),b) )= b_{\ell}^{tkn_{h}}$ concatenated with the $\delta_b$ application;
	 \item  $\delta'((\{tape_{\alpha_1},\dots, tape_{\alpha_i}, \dots\}, K,k, bt,f,P), {\sf read}())= \{b_0\}^\frown f(bt)$;
	\item  $\delta'((\{tape_{\alpha_1},\dots, tape_{\alpha_i}, \dots\}, K,k bt_0,f,P), {\sf read}())= b_0$.
	\end{itemize}
\end{definition}

\begin{definition}[$\mathfrak{R}(BT$-$ADT, \Theta_P)$  refinement] Same definition as the $\mathfrak{R}(BT$-$ADT, \Theta_F)$  refinement.
\end{definition}

%

\begin{definition}[k-Fork Coherence]
	A concurrent history $H=\langle \Sigma, E, \Lambda, \mapsto, \prec, \nearrow \rangle$ of the BT-ADT  composed with $\Theta_F$-ADT satisfies the \emph{k-Fork Coherence} if there are at most $k$ {\sf append}$()$ operations that return $\top$ for the same token. 
\end{definition}

\begin{theorem}[k-Fork Coherence]\label{t:kFork}
	Each concurrent history $H=\langle \Sigma, E, \Lambda, \mapsto, \prec, \nearrow \rangle$ of the BT-ADT composed with  a $\Theta_F$-ADT satisfies the \emph{k-Fork Coherence}.
\end{theorem}

\begin{proof}
	We prove the theorem by considering the defined refinement (Definition \ref{def:refinement}) where  \emph{(i)} there are a infinite number of ${\sf getToken}()$ invocations for object $obj$ and \emph{(ii)} given a valid block as input parameter, the {\sf consumeToken}$()$ operation successfully terminates if it has been invoked less than $k$ times for the same token. 
	From the properties of the pseudo random sequences of tapes, if there are an infinite number of ${\sf getToken}()$ invocations for object $obj$ then there exists at least one response for which ${\sf getToken}()$ operation returns a token $t$, which, when passed as input of the {\sf consumeToken}$()$ operation it successfully terminates if at most $k-1$ tokens $t$ have been already consumed. 	
\end{proof} 

{Let us notice, the $\Theta_F$-ADT guarantees by construction the safety property (Theorem \ref{t:kFork}). Liveness properties (i.e., the Termination) for $\Theta_F$-ADT and $\Theta_P$-ADT depend on the communication model and failure model in which those are implemented.
}
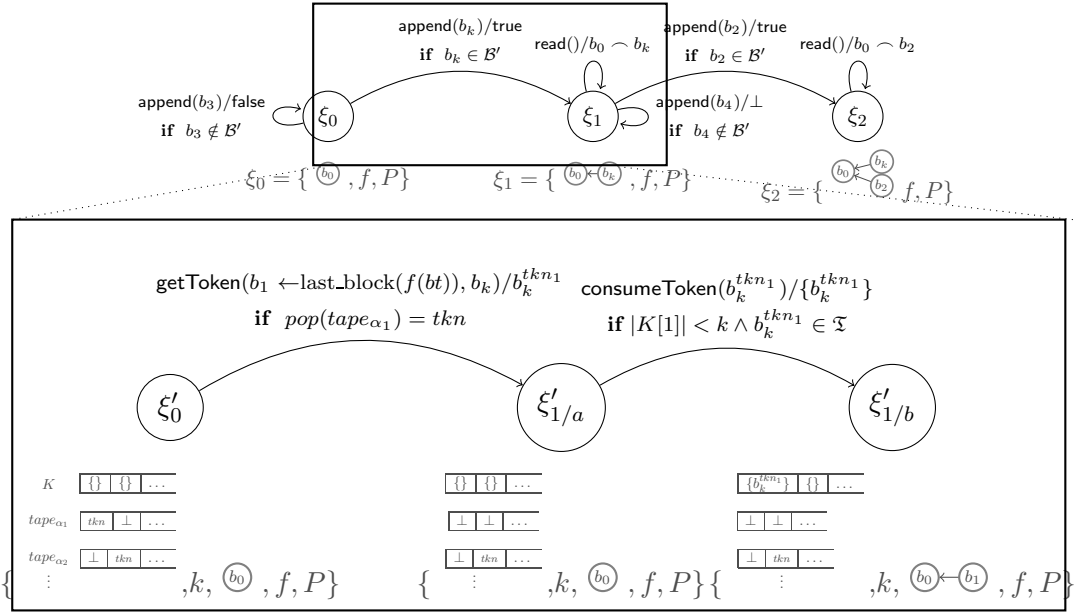
\begin{figure}
	\centering
	\begin{tikzpicture}[->, auto]
	\tikzset{stato/.style={draw,circle, node distance=4.4cm}}
	\tikzset{dettaglio/.style={black!70},node/.style={scale=0.6}}
	\tikzset{blocco/.style={
			circle,
			inner sep=1pt,
			thick,
			align=center,
			draw=gray,
		}
	}	
	\node[] () at (5,4) {
		\begin{tikzpicture}[->, auto, scale=0.5, every node/.style={scale=0.8}]

		\node[stato] (0) at (0,0) {$\xi_0 $};
		
		\node[dettaglio, below of=0] (s0) {$\xi_0 = \{$
			\begin{tikzpicture}[scale=0.8, every node/.style={scale=0.8}]	
			\node  [blocco] (i11) at (0.5,3)  {\scriptsize $b_0$};	
			\end{tikzpicture}
			$,f,P\}$};

		\node[stato, right of=0] (1) {$\xi_1$}; 
		\node[dettaglio, below of=1] (s1) {$\xi_1= \{$			
			\begin{tikzpicture}[scale=0.8, every node/.style={scale=0.8}]	
			\node  [blocco] (i11) at (0.5,3)  {\scriptsize $b_0$};
			\node  [blocco]  (i12) at (1.2,3) {\scriptsize $b_k$};
			\draw[<-] (i11) edge (i12)  ;
			\end{tikzpicture}
			$ ,f,P \}$};
		
		\path (0) edge [bend left] node[align=center] {\scriptsize {\sf append}$( b_k)/ {\sf true}$ \\ {\scriptsize{\bf if } $b_k \in \mathcal{B'} $}} (1);
		\path (0) edge [loop left] node[align=center] {\scriptsize {\sf append}$(b_3)/ {\sf false}$ \\ {\scriptsize{\bf if } $b_3 \notin \mathcal{B'} $}} (0);

		\node[stato, right of=1] (2) {$\xi_2$}; 
		\node[dettaglio, below of=2] (21) {$\xi_2= \{$
			\begin{tikzpicture}[scale=0.8, every node/.style={scale=0.8}]
			\node  [blocco] (i11) at (0.5,3)  {\scriptsize $b_0$};
			\node  [blocco]  (i12) at (1.3,3.2) {\scriptsize $b_k$};
			\node [blocco] (i13) at (1.3,2.7) {\scriptsize $b_2$};
			\draw[<-] (i11) edge (i12) (i11) edge (i13) ;
			\end{tikzpicture}
			$ f,P \}$};	
		
		\path (1) edge [bend left] node[align=center] {\scriptsize {\sf append}$( b_2)/ {\sf true}$ \\ {\scriptsize{\bf if } $b_2 \in \mathcal{B'} $}} (2);
		\path (1) edge [loop right] node[align=center] {\scriptsize {\sf append}$(b_4)/ {\bot}$ \\ {\scriptsize{\bf if } $b_4 \notin \mathcal{B'} $}} (1);

		\path (1) edge [loop above] node[align=center] {\scriptsize {\sf read}$()/ {b_0 \frown b_k}$} (1);	
		\path (2) edge [loop above] node[align=center] {\scriptsize {\sf read}$()/ {b_0 \frown b_2}$} (2);	
		
		\draw[thick] (-.4,-1.3) rectangle (9,3);	
	\end{tikzpicture}
};

\node[stato] (0) at (0,0) {$\xi_0^\prime $};

\node[dettaglio, below of=0] (s0) {$\{$
	\begin{tikzpicture}[scale=0.5, every node/.style={scale=0.5}]	
	\node[] () at (0,7){};
	\node[]() at (-.9,4){
		\begin{tabular}{|c|c|c}
		\hline
		$\{\}$&$\{\}$&$\dots$\\
		\hline
		\end{tabular}
	};
	\node[]() at (-3,4) {$K$};
	
	\node[]() at (-.9,3){
		\begin{tabular}{|c|c|c}
		\hline
		{\scriptsize $tkn$}&$\bot$&$\dots$\\
		\hline
		\end{tabular}
	};
	\node[]() at (-3,2) {$tape_{\alpha_2}$};
	
	\node[]() at (-.9,2){
		\begin{tabular}{|c|c|c}
		\hline
		$\bot$&{\scriptsize $tkn$}&$\dots$\\
		\hline
		\end{tabular}
	};
	\node[]() at (-3,3) {$tape_{\alpha_1}$};
	\node[]() at (-3,1.5) {$\vdots$};	
	\end{tikzpicture},$k$,
	\begin{tikzpicture}[scale=0.8, every node/.style={scale=0.8}]
	\node  [blocco] (i11) at (0.5,3)  {\scriptsize $b_0$};
	\end{tikzpicture}
	$,f,P\}$};

\node[stato] (1) at (5.2,0) {$\xi_{1/a}^\prime$}; 
\node[dettaglio, below of=1] (s1) {$ \{$			
	\begin{tikzpicture}[scale=0.5, every node/.style={scale=0.5}]	
	\node[] () at (-2,7){};
	\node[]() at (-2.5,4){
		\begin{tabular}{|c|c|c}
		\hline
		$\{\}$&$\{\}$&$\dots$\\
		\hline
		\end{tabular}
	};
	
	\node[]() at (-2.5,3){
		\begin{tabular}{|c|c|c}
		\hline
		{$\bot$}&$\bot$&$\dots$\\
		\hline
		\end{tabular}
	};
	
	\node[]() at (-2.5,2){
		\begin{tabular}{|c|c|c}
		\hline
		$\bot$&{\scriptsize $tkn$}&$\dots$\\
		\hline
		\end{tabular}
	};
	\node[]() at (-3,1.5) {$\vdots$};	
	\end{tikzpicture},$k$,
	\begin{tikzpicture}[scale=0.8, every node/.style={scale=0.8}]
	\node  [blocco] (i11) at (0.5,3)  {\scriptsize $b_0$};
	\end{tikzpicture}
	$,f,P\}$};
\path (0) edge [bend left] node[align=center] {\scriptsize {\sf getToken}$(b_1\leftarrow$last\_block$(f(bt)), b_k)/{b_k^{tkn_1}}$ \\ {\scriptsize{\bf if } $pop(tape_{\alpha_1})={tkn}$} } (1);


\node[stato, right of=1] (2) {$\xi_{1/b}^\prime$}; 
\node[dettaglio, below of=2] (21) {$ \{$
	\begin{tikzpicture}[scale=0.5, every node/.style={scale=0.5}]	
	\node[] () at (-4,7){};
	\node[]() at (-2.5,4){
		\begin{tabular}{|c|c|c}
		\hline
		$\{b_k^{tkn_1}\}$&$\{\}$&$\dots$\\
		\hline
		\end{tabular}
	};
	
	\node[]() at (-3,3){
		\begin{tabular}{|c|c|c}
		\hline
		{$\bot$}&$\bot$&$\dots$\\
		\hline
		\end{tabular}
	};
	
	\node[]() at (-2.9,2){
		\begin{tabular}{|c|c|c}
		\hline
		$\bot$&{\scriptsize $tkn$}&$\dots$\\
		\hline
		\end{tabular}
	};
	\node[]() at (-3,1.5) {$\vdots$};	
	\end{tikzpicture},$k$,
	\begin{tikzpicture}[scale=0.8, every node/.style={scale=0.8}]
	\node  [blocco] (i11) at (0.5,3)  {\scriptsize $b_0$};
	\node  [blocco]  (i12) at (1.3,3) {\scriptsize $b_1$};
	\draw[<-] (i11) edge (i12)  ;
	\end{tikzpicture}
	$,f,P \}$};	

\path (1) edge [bend left] node[align=center] {\scriptsize {\sf consumeToken}$(b_k^{tkn_1})/\{b_k^{tkn_1}\}$ \\ {\scriptsize {\bf if} $|K[1]|<k \wedge b_k^{tkn_1} \in \mathfrak{T}$}  }  (2);

\draw[thick] (-2.1,-2.7) rectangle (11.9,2.5); 
\draw[-, dotted] (2,3.2)-- (-2.1,2.5);
\draw[-, dotted] (12,2.5)-- (6,3.2);
\end{tikzpicture}

\caption{Refinement of the {\sf append}$()$ operation. We use the following syntax on the edges: {\sf operation}/{\sf output}.}\label{fig:Composition}
\end{figure}

%

\subsection{Hierarchy}
In this section we define a hierarchy between different BT-ADT satisfying different consistency criteria when augmented with different oracle ADT. We use the following notation:  BT-ADT$_{SC}$ and BT-ADT$_{E C}$ to refer respectively to BT-ADT generating concurrent histories that satisfies the $SC$ and the $E C$ consistency criteria. When augmented with the oracles we have the following four typologies, where for the \emph{frugal} oracle we explicit the value of $k$: $\mathfrak{R}(\text{BT-ADT}_{SC}, \Theta_{F,k})$, $\mathfrak{R}(\text{BT-ADT}_{SC}, \Theta_P)$, $\mathfrak{R}(\text{BT-ADT}_{ E C}, \Theta_{P})$, $\mathfrak{R}(\text{BT-ADT}_{E C}, \Theta_{F,k})$.\\

In the following we want study the relationship among the different refinements. Without loss of generality, let us consider only the set of histories $\hat{\mathcal{H}}^{\mathfrak{R}(\text{BT-ADT},\Theta)}$ such that each history $\hat{H}^{\mathfrak{R}(\text{BT-ADT},\Theta)} \in \hat{\mathcal{H}}^{\mathfrak{R}(BT-ADT,\Theta)}$  is purged from the unsuccessful {\sf append}$()$ response events (i.e., such that the returned value is $\bot$). Let $\hat{\mathcal{H}}^{\mathfrak{R}(\text{BT-ADT},\Theta_{F,k})}$ be the concurrent set of histories generated by a BT-ADT refined with $\Theta_{F,k}$-ADT and let $\hat{\mathcal{H}}^{\mathfrak{R}(\text{BT-ADT},\Theta_P)}$ be the concurrent set of histories generated by a BT-ADT refined with $\Theta_P$-ADT.

\begin{theorem}\label{t:oracolsHierarchy}
	$\hat{\mathcal{H}}^{\mathfrak{R}(\text{BT-ADT},\Theta_F)} \subseteq \hat{\mathcal{H}}^{\mathfrak{R}(\text{BT-ADT},\Theta_P)}$. 
\end{theorem}

\begin{proof}
	The proof follows from Theorem \ref{t:kFork} considering that $\mathfrak{R}(BT,\Theta_P)$ can generate histories with an infinite number of {\sf append}$()$ operations that successfully terminate while $\mathfrak{R}(BT,\Theta_F)$ can generate history with at most $k$ {\sf append}$()$ operations that successfully terminate.
\end{proof}

\begin{theorem}\label{t:oracolsHierarchyF}
	If $k_1 \leq k_2$ then $\hat{\mathcal{H}}^{\mathfrak{R}(\text{BT-ADT},\Theta_{F,k_1})} \subseteq \hat{\mathcal{H}}^{\mathfrak{R}(\text{BT-ADT},\Theta_{F,k_2})}$. 
\end{theorem}

\begin{proof}
	The proof follows from Theorem \ref{t:kFork} applying the same reasoning as for the proof of Theorem \ref{t:oracolsHierarchy} with $k_1 \leq k_2$.
\end{proof}

Finally, from Theorem \ref{th:strongImpliesEventual} the next corollary follows.

\begin{corollary}\label{c:oracolsHierarchySameOracle}
	$\hat{\mathcal{H}}^{(\mathfrak{R}(\text{BT-ADT}_{SC},\Theta)} \subseteq \hat{\mathcal{H}}^{\mathfrak{R}(\text{BT-ADT}_{E C},\Theta)}$. 
\end{corollary}


Combining Theorem \ref{th:strongImpliesEventual} and Theorem \ref{t:oracolsHierarchy} we obtain the hierarchy depicted in Figure \ref{fig:BTHierarchy}.

%
%
%
%
%
%

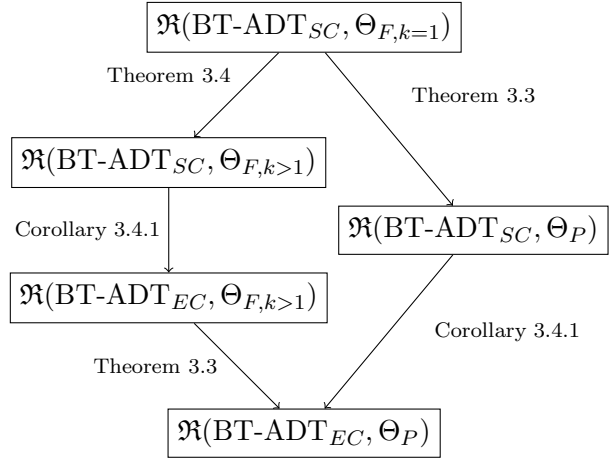
\begin{figure}
	\centering
	\begin{tikzpicture}[scale=0.9]
	
	\node[draw](1) at (0,2) {$\mathfrak{R}(\text{BT-ADT}_{SC}, \Theta_{F,k=1})$};
	\node[draw](2) at (-2,-2) {$\mathfrak{R}(\text{BT-ADT}_{ E C}, \Theta_{F,k>1})$};
	
	\node[draw](3) at (2.5,-1) {$\mathfrak{R}(\text{BT-ADT}_{ SC}, \Theta_{P})$};
	\node[draw](5) at (-2,0) { $\mathfrak{R}(\text{BT-ADT}_{ SC}, \Theta_{F,k>1})$};
	\node[draw](4) at (0,-4) {$\mathfrak{R}(\text{BT-ADT}_{ E C}, \Theta_{P})$};
	
	
	\draw[->]  (2) edge (4) (5) edge (2) 	(1) edge (3)
	(1) edge (5)	(3) edge (4) ;
	
	
	\node[] () at (-2, 1.3) {\scriptsize Theorem \ref{t:oracolsHierarchyF}};
	\node[] () at (2.5, 1) {\scriptsize Theorem \ref{t:oracolsHierarchy}};
	\node[] () at (-2.2, -3) {\scriptsize Theorem \ref{t:oracolsHierarchy}};
	\node[] () at (-3.2, -1) {\scriptsize Corollary \ref{c:oracolsHierarchySameOracle}};
	\node[] () at (3, -2.5) {\scriptsize Corollary \ref{c:oracolsHierarchySameOracle}};			
	\end{tikzpicture}
	\caption{$\mathfrak{R}(\text{BT-ADT}, \Theta)$ Hierarchy.}\label{fig:BTHierarchy}
\end{figure}

 \label{sec:composition}



\section{
	Implementing BT-ADTs}

\subsection{Implementability in a concurrent model}
In this Section we show that $\Theta_{F,k=1}$ has consensus number $\infty$ and that $\Theta_{P}$ 
has consensus number 1.

We consider a concurrent system composed by $n$ processes such that up to $f$ processes are faulty (stop prematurely by crashing), $f<n$. Moreover, processes can communicate through atomic registers.



\subsubsection{Frugal with $k=1$ at least as strong as Consensus}


In the following we prove that there exists a wait-free implementation of the Consensus \cite{lamport1982byzantine} by the $\Theta_{F,k=1}$ Oracle object.
In particular, in this case $\Theta_{F,k=1}$ = $\langle$ A= \{{\sf getToken}$(b_h,b_\ell)$, {\sf consumeToken}$( b_{\ell}^{tkn_{h}}):
b_h, b_{\ell}^{tkn_{h}} \in \mathcal{B}^\prime, b_\ell \in \mathcal{B}, tkn_h \in \mathfrak{T}\}$, B= $\mathcal{B}^\prime  \cup Boolean$, Z= $m(\mathcal{A})^* \times \{K\} \times k \cup \{pop, head, dec, get\}$, $\xi_0, \tau, \delta \rangle$. We explicit consider blocks and valid blocks ($\mathcal{B}$ and $\mathcal{B}^\prime$) rather than objects and valid objects ($\mathcal{O}$ and $\mathcal{O}^\prime$). Moreover, we consider a version of the Consensus problem for the blockchain. Thus, we consider the Validity property as in \cite{redbelly17} such that the decided block $b$ satisfies the predicate $P$. 
\begin{definition}
	Consensus $\mathcal{C}$:
	\begin{itemize}
		\item \textbf{Termination.} Every correct process eventually decides some value.
		\item \textbf{Integrity.} No correct process decides twice.
		\item \textbf{Agreement.} If there is a correct process that decides a value $b$, then eventually all the correct processes decide $b$.
		\item \textbf{Validity\cite{redbelly17}.} A decided value is valid, it satisfies the predefined predicate denoted $P$.
	\end{itemize}
\end{definition}

To this aim, we first prove that there exists a wait-free implementation of {\sf Compare\&Swap}$()$ object by {\sf consumeToken}$()$ object in the case of $\Theta_{F,k=1}$, implying that {\sf consumeToken}$()$ has the same Consensus number as {\sf Compare\&Swap}$()$ which is $\infty$ (see \cite{herlihy1991wait}). Finally we compose the {\sf consumeToken}$()$ with the {\sf getToken}$()$ object proving that there exist a wait-free implementation of $\mathcal{C}$ by $\Theta_{F,k=1}$.

Figure \ref{fig:CAS} describes {\sf consumeToken}$()$ (CT), as specified by the $\Theta$-ADT, along with the {\sf Compare\&Swap}$()$ (CAS). {\sf Compare\&Swap}$()$ takes three parameters as input, the $register$, the $old\_value$ and the $new\_value$. If the value in $register$ is the same as  $old\_value$ then the $new\_value$ is stored in $register$ and in any case the operation returns the value that was in $register$ at the beginning of the operation. 
In comparison with {\sf consumeToken}$(b_\ell^{tkn_h})$ we have that $b_\ell^{tkn_h}$ is the $new\_value$, $register$ is $K[h]$ and the implicit $old\_value$ is $\{\}$. That is, $add(K,h,b)$ stores $b$ in $K[h]$ if $|K[h]|<k=1$, then if $K[h]=\{\}$. In any case the operation returns the content of $K[h]$ at the end of the operation itself. Figure \ref{fig:CAStoCT} describes and algorithm that reduces CAS to {\sf consumeToken}$()$.

\begin{figure*}[t]
	\centering
	\fbox{
		\begin{minipage}{0.4\textwidth}
			\scriptsize
			\resetline
			\begin{tabbing}
				aaaA\=aA\=aA\=aaaA\kill
				\line{} {\sf consumeToken}$(b_\ell^{tkn_h}):$\\
				\line{} \> $previous\_value \leftarrow K[h];$\\
				\line{} \> {\bf if} $(previous\_value == \{\} \wedge tkn_h \in \mathfrak{T})${\bf then};\\
				\line{} \>\>$K[h] \leftarrow K[h] \cup \{b_\ell^{tkn_h}\}$;\\
				\line{} \>\> {\bf endIf}\\
				\line{} \> {\bf return} $K[h]$\\
			\end{tabbing}
			\normalsize
		\end{minipage}%
		\vspace{.5cm}
		\begin{minipage}{0.4\textwidth}
			\scriptsize
			\resetline
			\begin{tabbing}
				aaaA\=aA\=aA\=aaaA\kill
				\line{} {\sf compare\&swap}$(register, old\_value, new\_value):$\\
				\line{} \> $previous\_value \leftarrow register;$\\
				\line{} \> {\bf if} $(previous\_value == old\_value)${\bf then};\\
				\line{} \>\>$register \leftarrow new\_value$;\\
				\line{} \>\> {\bf endIf}\\
				\line{} \> {\bf return} $previous\_value$\\
			\end{tabbing}
			\normalsize
		\end{minipage}
		
	}
	\caption{{\sf Compare\&Swap}$()$ and {\sf consumeToken}$()$ in the case of $\Theta_{F,k=1}$.}
	\label{fig:CAS}   
\end{figure*}

\begin{figure*}[t]
	\centering
	\fbox{
		\begin{minipage}{0.4\textwidth}
			\scriptsize
			\resetline
			\begin{tabbing}
				aaaA\=aA\=aA\=aaaA\kill
				\line{} {\sf compare\&swap}$(K[h], \{\}, b_\ell^{tkn_h}):$\\
				\line{}\> $returned\_value\leftarrow${\sf consumeToken}$(b_\ell^{tkn_h});$\\
				\line{CAStoCT-03} \> {\bf if} $(returned\_value == b_\ell^{tkn_h})${\bf then};\\
				\line{} \>\>{\bf return} $\{\}$;\\
				\line{} \>\> {\bf else} {\bf return} $returned\_value$;\\
				\line{} \>\> {\bf endIf} \\
			\end{tabbing}
			\normalsize
		\end{minipage}	
	}
	\caption{An implementation of CAS by CT in the case of $\Theta_{F,k=1}$.}
	\label{fig:CAStoCT}   
\end{figure*}

\begin{theorem}\label{t:CAStoCT}
	If input values are in $\mathcal{B'}$ then there exists an implementation of CAS by CT in the case of $\Theta_{F,k=1}$.
\end{theorem}

\begin{proof}
	The proof simply follows by construction. Let us consider the algorithm in Figure \ref{fig:CAStoCT}. When the {\sf Compare\&Swap}$()$ operation is invoked, if $K[h]$ is empty, then when {\sf consumeToken}$()$ is invoked with $b_\ell^{tkn_h}$ (valid by hypothesis) $K[h]$ is populated with $b_\ell^{tkn_h}$. Such value is later returned by the {\sf consumeToken}$()$ operation in $retuend\_value$. Since it is the same value as $b_\ell^{tkn_h}$ (line \ref{CAStoCT-03}) then the {\sf Compare\&Swap} returns the value of $K[h]$ at the beginning of the operation, $\{\}$. If the condition at line (line \ref{CAStoCT-03}) does not hold, then this means that $K[h]$ did not change during the operation and its value, in $returned\_value$ is returned.
\end{proof}

Figure \ref{fig:reduction} describes a simple implementation of Consensus by $\Theta_{F,k=1}$. When a correct process $p_i$ invokes the {\sf propose}$(b)$ operation it loops invoking the {\sf getToken}$(b_0, b)$ operation as long as a valid block is returned (lines \ref{r:while}-\ref{r:gettoken}). In this case the {\sf getToken}$()$ operation takes as input some block $b_0$ and the proposed block $b$. Afterwards, when the valid block has been obtained $p_i$ invokes the {\sf consumeToken}$(validBlock)$ operation whose result in stored in the $tokenSet$ variable (line \ref{r:consumetoken}). Notice, the first process that invokes such operation is able to successfully consume the token, i.e., the valid block is in the Oracle set corresponding to $b_0$, which cardinality is $k=1$, and such set is returned each time the {\sf consumeToken}$()$ operation is invoked for a block related to $b_0$. Finally, (line \ref{r:decision}) the decision is triggered on such set (with contains one element).

\begin{figure*}[t]
	\centering
	\fbox{
		\begin{minipage}{0.4\textwidth}
			\scriptsize
			\resetline
			\begin{tabbing}
				aaaA\=aA\=aA\=aaaA\kill
				
%
%
				
				{\bf upon event} ${\sf propose}(b)$:\\
				\line{} \> $validBlock \leftarrow \bot;$\\
				\line{} \> $validBlockSet \leftarrow \emptyset;$ \% since $k=1$ then it contains only one element. \\
				\line{r:while} \> {\bf while} $(validBlock = \bot)$: \\
				\line{r:gettoken} \>\>$validBlock \leftarrow \sf{getToken}(b_0, b)$;\\
				\line{r:consumetoken} \>$validBlockSet \leftarrow {\sf consumeToken}(validBlock)$; \% it can be different from validBlock\\
				\line{r:decision} \> {\bf trigger} {\sf decide}$(validBlockSet)$;\\

			\end{tabbing}
			\normalsize
		\end{minipage}%
	}
	\caption{The Protocol $\mathcal{A}$ that reduces the Consensus problem to the Frugal Oracle with $k=1$.}
	\label{fig:reduction}   
\end{figure*}

\begin{theorem}\label{t:consensusOracleReduction}
	$\Theta_{F,k=1}$ Oracle has Consensus number $\infty$.
\end{theorem}

\begin{proof}
	The proof proceeds by construction, let us consider the implementation in Figure \ref{fig:reduction}. All correct processes performing the Consensus are looping on the {\sf getToken}$(b_0,b)$ operation. From the properties of the pseudo random sequences of tapes, if there are an infinite number of ${\sf getToken}()$ invocations for an block $b_0$ then there exists at least one response for which ${\sf getToken}()$ operation returns a valid block $b^{tkn_0}$.
	Thus, all correct process $i$ can invoke the {\sf consumeToken}$(b^{tkn_0})$ operation with valid values. Since all the processes invoke such operation with valid values with can apply Theorem \ref{t:CAStoCT} which concludes the proof considering that CAS has Consensus number $\infty$ (\cite{herlihy1991wait}).
\end{proof}

 \label{ssec:reductionConsensus}

\subsubsection{Prodigal not stronger than an Atomic Register}


In order to show that the Prodigal oracle $\Theta_{P}$ has consensus number 1, it suffices to find a wait-free implementation of the oracle by an object with consensus number 1. To this end we present a straightforward implementation of the Prodigal oracle by Atomic Snapshot\cite{AH1990waitfree}.

Let us firstly simplify the notation of the consume token operation. Let us consider a consume token invoked for a given block $b_h$, denoted as ${\sf consumeToken_h}(tkn_{m})$, which simply writes a token from the set $\mathfrak{T}=\{tkn_1,tkn_2, \dots, tkn_m,\dots\}$ in the set $K[h]$. Without loss of generality let us assume that: (i) tokens are  uniquely identified , (ii)  cardinality of $\mathfrak{T}$ is $n$ finite but not known and (iii) the set $K[h]$ is  represented by a collection of $n$ atomic registers $\mathfrak{K[h}=\{R_{h,1},R_{h,2}, \dots, R_{h,m},\dots R_{h,n}\}$, where $R_{h,m}$ is assigned to the $tkn_m$ token, i.e. $R_{h,m}$ can contain either $\bot$ or $tkn_m$. 

It can be observed that the  ${\sf consumeToken_h}(tkn_{m})$ in the case of $k$ infinite, always allows to write the token $tkn_{m}$ in $R_{h,m}$, i.e. there always exists a register $R_{h,m}$ for the proposed token $tkn_{m}$. By the oracle definition, moreover, the ${\sf consumeToken_h}(tkn_{m})$ returns a read of the $n$ registers that includes the last written token. Figure \ref{fig:CTtoAS} shows a  trivial implementation of commit token $CT$ using Atomic Snapshot that offers  ${\sf update}(R_i, value)$,  ${\sf scan}(R_1, R_2, \dots, R_n)$ operation to update a particular register and perform an atomic read of input registers, respectively. 

\begin{figure*}[t]
	\centering
	\fbox{
		\begin{minipage}{0.4\textwidth}
			\scriptsize
			\resetline
			\begin{tabbing}
				aaaA\=aA\=aA\=aaaA\kill
				\line{} {\sf consumeToken$_k$}$(tkn):$\\
				\line{}\> \> $R_{h,m}$ $\leftarrow${\sf update}$(R_{h,m}, tkn_m)$\\
				\line{} \>\>  $returned\_value \leftarrow {\sf scan(R_{h,1},R_{h,2}, \dots, R_{h,m},\dots R_{h,n})}$\\
				\line{}  \>\>{\bf return} $returned\_value$;\\
			\end{tabbing}
			\normalsize
		\end{minipage}	
	}
	\caption{An implementation of CT by Atomic Snapshot in the case of $\Theta_{P}$.}
	\label{fig:CTtoAS}   
\end{figure*}

\begin{theorem}\label{t:CTtoAS}
	$\Theta_{P}$ Oracle has Consensus number $1$.
\end{theorem}

\begin{proof}
	The proof trivially follows from  implementation  in Figure \ref{fig:CTtoAS}  of the consensus token operation of the Prodigal oracle by the Atomic Snaposhot object and from \cite{AH1990waitfree}. 
\end{proof} \label{ssec:reductionRegister}

\subsection{Implementability in a message-passing system model}
We consider a message-passing system composed of an arbitrary large but finite set of $n$ processes, $\Pi=\{p_1, \dots, p_n\}$.
The passage of time is measured by a fictional global clock (\emph{e.g.}, that spans the set of natural integers). Processes in the system do not have access to the fictional global time.
Each process of the distributed system executes a single instance of a distributed protocol $\mathcal{P}$ composed of a set of algorithms, i.e., each process is running an algorithm. Processes can exhibit a Byzantine behavior  (i.e., they can arbitrarily deviate from the protocol $\mathcal{P}$ they are supposed to run). A process affected by a Byzantine behavior is said to be faulty, otherwise we refer to such process as non-faulty or correct. We make no assumption on the number of failures that can occur during the system execution. Processes communicate by exchanging messages via communication channels. We say that a communication channels are asynchronous if the is no upper bound on the message delivery delay. Contrarily, communication channels are synchronous if messages sent by correct processes at time $t$ are delivered by correct processes by time $t+\delta$. Finally, communication channels are weakly synchronous if there exist an unknown a priori time $\tau$ after which the communication channels behave as synchronous.  We specify time to time the channels synchrony assumption considered, when left untold we consider asynchronous channels.

The BlockTree being now a shared object replicated at each process, we note by  $bt_i$ the local copy of the BlockTree maintained at  process $i$.
To maintain the replicated object we consider histories made of events related to the {\sf read}  and {\sf append}  operations on the shared object,  i.e. the {\sf send} and {\sf receive} operations for process communications and  the {\sf update} operation for BlockTree updates.  We also use subscript  $i$ to indicate that the operation occurred at process $i$: {\sf update}$_i(b_g,b_i$) indicates that $i$ inserts its locally generated valid block \emph{b}$_i$ in $bt_i$ with \emph{$b_g$} as a predecessor. Updates are communicated through {\sf send} and {\sf receive} operations. An update related to a block $b_i$ generated on a process $p_i$, sent through {\sf send}$_i(b_g,b_i)$, and received through a {\sf receive}$_j(b_g,b_i)$, takes effect on the local replica $bt_j$ of $p_j$ with the operation {\sf update}$_j(b_g,b_i)$. 

We assume a generic implementation of the {\sf update} operation: when process $i$ locally updates its BlockTree $bt_i$  with the valid block $b_i$ (returned from the {\sf consumeToken()} operation), we write   {\sf update}$_i(b,b_i)$. When a process $j$ execute the {\sf receive}$_j(b,b_i)$ operation, it locally updates its BlockTree $bt_j$ by invoking the {\sf update}$_j(b,b_i)$ operation.  

In the remaining part of the work we consider implementations of BT-ADT  in a Byzantine failure model where the set of events is restricted as follows.
\begin{definition}
	The execution of the system that uses the BT-ADT =(A, B, Z, $\xi_0, \tau, \delta$) in a Byzantine failure model defines the concurrent history  $H=\langle \Sigma, E, \Lambda, \mapsto, \prec, \nearrow \rangle$ (see Definition \ref{def:adthistory}) where we restrict $E$ to a countable set of events that contains \emph{(i)} all the BT-ADT  {\sf read}$()$ operations invocation events by the \emph{correct} processes,  \emph{(ii)} all BT-ADT {\sf read}$()$ operations response events at the \emph{correct} processes,  \emph{(iii)} all {\sf append}$(b)$ operations invocation events such that $b$ satisfies the predicate $P$ and,  \emph{(iv)} $send$, $receive$ and $update$ events generated at correct processes.
\end{definition}

\label{sec:system}
In this Section we consider a message passing system model and we show the

(i) impossibility to achieve Strong Prefix without Consensus and impossibility to achieve Eventual Prefix if at least one message sent by a correct process is lost. 

TBC: 
(ii)Eventual Prefix is impossible in an asynchronous system
(iii)Eventual Prefix is impossible if the interval between the generation of two successive blocks is less than the upper bound on the message delay. 
(iv) Impossible to solve Strong Prefix without the Frugal oracle with $k=1$.

\subsection{Communication Abstractions}
\label{sec:comm-abstr-bt}

We now define the properties that each history $H$ generated by a BT-ADT satisfying the Eventual Prefix Property has to satisfy and then we prove their necessity.

\begin{definition}[Update Agreement]\label{def:historyUpdate}
	A concurrent history $H=\langle \Sigma, E, \Lambda, \mapsto, \prec, \nearrow \rangle$ of the system that uses a BT-ADT satisfies the Update Agreement if satisfies the following  properties:
	\begin{itemize}
		\item{R1.} $\forall {\sf update}_i(b_g,b_i) \in H$,$\exists {\sf send}_i(b_g,b_i)\in H$;
		\item{R2.} $\forall {\sf update}_i(b_g,b_j) \in H, \exists {\sf receive}_i(b_g,b_j) \in H$ such that ${\sf receive}_i(b_g,b_j) \mapsto {\sf update}_i(b_g,b_j)$;
		\item{R3.} $\forall {\sf update}_i(b_g,b_j) \in H$, $\exists {\sf receive}_k(b_g,b_j)\in H, \forall k$.
	\end{itemize}
\end{definition}
Figure \ref{fig:run} depicts a concurrent history that satisfies the Update Agreement properties. 
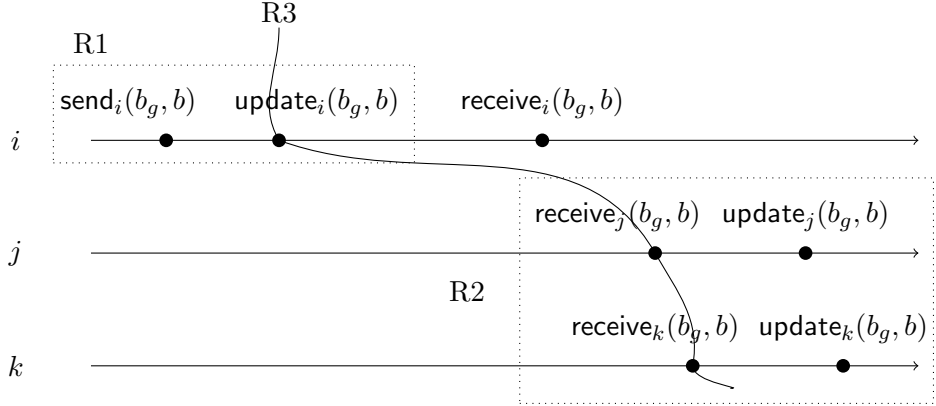
\begin{figure}
	\centering
	\begin{tikzpicture}
	\tikzset{evento/.style={circle, scale=0.5, fill=black}}
	\draw[->] (0,0) -- (11,0); \node[] () at (-1,0) {$k$};
	\draw[->] (0,1.5) -- (11,1.5); \node[] () at (-1,1.5) {$j$};
	\draw[->] (0,3) -- (11,3);\node[] () at (-1,3) {$i$};
	
	\node[]() at (.5,3.5) {{\sf send}$_i(b_g,b)$};
	\node[evento]() at (1,3) {};
	
	\node[]() at (3,3.5) {{\sf update}$_i(b_g,b)$};
	\node[evento]() at (2.5,3) {}; 
	
	\node[]() at (6,3.5) {{\sf receive}$_i(b_g,b)$};
	\node[evento]() at (6,3) {};
	
	\node[]() at (7,2) {{\sf receive}$_j(b_g,b)$};
	\node[evento]() at (7.5,1.5) {}; 
	
	\node[]() at (7.5,.5) {{\sf receive}$_k(b_g,b)$};
	\node[evento]() at (8,0) {}; 
	
	
	\node[]() at (9.5,2) {{\sf update}$_j(b_g,b)$};
	\node[evento]() at (9.5,1.5) {}; 
	
	\node[]() at (10,.5) {{\sf update}$_k(b_g,b)$};
	\node[evento]() at (10,0) {};
	
	\draw   (2.5,4.5) to[out=270,in=120] (2.5,3) to[out=-20,in=120] (7.5,1.5)  to[out=300,in=80] (8,0) to[out=270,in=0] (8.5,-.3); 
	\node[] () at (2.5,4.7) {R3};
	\draw[dotted] (-.5,2.7) rectangle (4.3,4); \node[] () at (0,4.3) {R1};
	\draw[dotted] (5.7,2.5) rectangle (11.2,-.5); \node[] () at (5,1) {R2};
	
	\end{tikzpicture}
	\caption{Example of concurrent history that satisfies R1,R2 and R3, the Update Agreement properties. }\label{fig:run}
\end{figure}

In the following, for ease of notation we consider that the selection function $f \in \mathcal{F}$ returns directly also the genesis block.

\begin{lemma}\label{l:historyConditions12}
	Property R1 or Property R2 are necessary conditions for any protocol $\mathcal{P}$ to implement a BT-ADT generating histories $H$ satisfying the Eventual Prefix property.
\end{lemma}

\begin{proof}
	Let us assume that there exists a protocol $\mathcal{P}$ implementing a BT-ADT that generates histories $H$ satisfying Eventual Prefix property but not Property R1 or Property R2. Thus, in $H$ there is some {\sf update} $u$ that is not sent to the other processes (R1) or once received, $u$ is not locally applied (R2). Let us consider the following history where R1 is not verified and process $i$ issues the first {\sf update} event in $H$. \\
	Let us construct the following execution history $H$. $i$ issues the {\sf update}$_i(b_0,b_i^\prime)$ (thus $bt_i=b_0 ^\frown b_i^\prime$) but not the {\sf send}$_i(b_0,b_i^\prime)$ event. It follows that if there is no {\sf send}$_i(b_0,b_i^\prime)$ event in $H$ then in $H$ are no present any {\sf receive}$_j(b_0,b)$ events, $j \neq i$ and thus not process $j \neq i$ can issue {\sf update}$_j(b_0,b_i^\prime)$ (on the other side, if R2 is not satisfied, even if the the {\sf receive}$_j(b_0,b)$ event occur then {\sf update}$_j(b_0,b_i^\prime)$ may not occur), thus $\forall j\neq i, bt_j=b_0$. Let us assume that $i$ performs a {\sf read}$()$ operation, the selection function $f\in  \mathcal{F}$ is applied on $bt_i=b_0 ^\frown b_i^\prime$. 
	By the {\sf score} function definition it follows that ${\sf score}(b_0 ^\frown b_i^\prime)>{\sf score}(b_0)$. Thus if $i$ issues a {\sf read}$()$ operation after {\sf update}$_i(b_0,b_i^\prime)$ it returns a blockchain such that ${\sf score}(b_0 ^\frown b)$ and the possible infinite {\sf read}$()$ operations issued by other processes always return blockchain such that ${\sf score}(b_0)$, violating the Eventual Prefix property. The construction of $H$ can be completed iterating the same reasoning for an infinite number of {\sf append}$()$ operation issued by $i$, thus $H$ violates the Eventual Prefix Property leading to a contradiction. 
\end{proof}

\begin{lemma}\label{l:historyCondition3}
	Property R3 is a necessary condition for any protocol $\mathcal{P}$ to implement a BT-ADT generating histories $H$ satisfying the Eventual Prefix property.
\end{lemma}

\begin{proof}
	Let us assume that there exists a protocol $\mathcal{P}$ implementing a BT-ADT that generates histories $H$ satisfying Eventual Prefix property but not Property R3. Thus, in $H$ there is some {\sf update}$_i(b,b_i^\prime)$ $u$ at some process $i$ such that the {\sf receive}$_j(b,b_i^\prime)$ events do not occur at all processes $j \neq i$. \\Let us consider a system composed by three processes, $i,j$ and $k$. The system execution generates the following history $H$ where R3 is not verified. In particular, in $H$ are present the {\sf update}$_i(b_0,b_i ^\prime)$, {\sf receive}$_j(b_0,b_i ^\prime)$ events but there is no any {\sf receive}$_k(b_0,b_i ^\prime)$ event. It follows that $bt_i=bt_j=b_0 \frown b_i^\prime$ and $bt_z=b_0$.
	We apply the same argument as for Lemma \ref{l:historyConditions12}. Let us assume that $j$ and $k$ perform {\sf read}$()$ operations. Such operation returns the result of $f(bt_j)$ and $f(bt_k)$ respectively. By the {\sf score} function definition it follows that ${\sf score}(b_0 ^\frown b_i^\prime)>{\sf score}(b_0)$. If $j$ issues a {\sf read}$()$ operation after {\sf update}$_j(b_0, b_i^\prime)$ it returns a blockchain with ${\sf score}(b_0 ^\frown b)$ and the other {\sf read}$()$ operations issued by $k$ will always return blockchain with ${\sf score}(b_0)$. The construction of $H$ can be completed iterating the same reasoning for an infinite number of {\sf append}$()$ operation issued by $i$, thus $H$ violates the Eventual Prefix Property leading to a contradiction. 
\end{proof}

\begin{theorem}\label{th:updateAgreementNecessityEventual}
	The update agreement property is necessary to construct concurrent histories $H=\langle \Sigma, E, \Lambda, \mapsto, \prec, \nearrow \rangle$ generated by a BT-ADT that satisfy the BT Eventual Consistency criterion.
\end{theorem}

\begin{proof}
	The proof follows directly from Lemma \ref{l:historyConditions12}, Lemma \ref{l:historyCondition3} and the definition of Eventual BT consistency criterion. 
\end{proof}

Considering Theorem \ref{th:updateAgreementNecessityEventual} and Theorem \ref{th:strongImpliesEventual} the next Corollary follows.

\begin{corollary}\label{c:updateAgreementNecessityEventual}
	There not exists a concurrent history $H=\langle \Sigma, E, \Lambda, \mapsto, \prec, \nearrow \rangle$ of the system that uses a BT-ADT that satisfies the Strong BT consistency criterion but not the Update Agreement.
\end{corollary}

%
%
%

In the following we consider a communication primitive that is inspired by the Liveness properties of the reliable broadcast \cite{SDbook2011}. We will prove that this abstraction is necessary to implement Eventual BT Consistency.
%

\begin{definition}[Light Reliable Communication (LRC)]
	A concurrent history $H$ satisfies the properties of the LRC abstraction if and only if:
	\begin{itemize}
		\item (Validity): 
		$\forall {\sf send}_i(b,b_i) \in H, \exists {\sf receive}_i(b,b_i) \in H$;
		\item (Agreement): 
		$\forall {\sf receive}_i(b,b_j) \in H, \forall k \exists {\sf receive}_k(b,b_i) \in H$
	\end{itemize}
\end{definition}
In other words, if a correct process $i$ sends a message $m$ then $i$ eventually receives $m$ and if a message $m$ is received by some correct process (e.g., $i$ itself), them $m$ is eventually received by every correct process.

\begin{theorem}\label{th:LRCNecessity}
	The LRC abstraction is necessary to for any BT-ADT implementation that generates concurrent histories that satisfies the BT Eventual Consistency criterion.
\end{theorem}

\begin{proof}
	The proof done by generating a concurrent history $H$ that violates the LRC properties and showing that $H$ also violate the Update Agreement properties.  For Theorem \ref{th:updateAgreementNecessityEventual} the Update Agreement properties are necessary condition to implement BT-ADT that generates concurrent histories that satisfies the BT Eventual Consistency criterion. \\
	Let us consider $H$ where at process $n$ occurs the event {\sf update}$_n(b,b_n)$ and {\sf send}$_n(b,b_n)$ and where the LRC2 property is not satisfied. 
	If LRC2 is violated then in $H$ we can have that there exist some process $i$ at which occurs the  ${\sf receive}_i(b,b_n)$ event and some process $j$ at which never occurs the ${\sf receive}_j(b,b_n)$ event. Since at process $n$ occurred the event {\sf update}$_n(b,b_n)$, then, for the R3 property, for each process $k$ {\sf update}$_n(b,b_n)$ has to occur. For R2 the {\sf update}$_m(b,b_n)$ event at some process $m$ has to be preceded by a {\sf receive}$_m(b,b_n)$ event at the same process $m$. Since by hypothesis not at all processes $m$ the {\sf receive}$_m(b,b_n)$ occurs then the property is violated, violating the Update Agreement properties, which are necessary conditions to implement BT-ADT that generates concurrent histories that satisfies the BT Eventual Consistency criterion, which concludes the proof.\\
\end{proof}

Finally, from Theorem \ref{th:strongImpliesEventual} and Theorem \ref{th:LRCNecessity} the next Corollary follows. 

\begin{corollary}\label{c:LRCNecessityStrong}
	The LRC abstraction is necessary to for any BT-ADT implementation that generates concurrent histories that satisfies the BT Strong Consistency criterion.
\end{corollary}

\subsection{System model and hierarchy}

\textbf{Observation.}
Following our Oracle based abstraction (Section \ref{sec:composition}) we assume by definition that the synchronization on the block to append is oracle side and takes place during the append operation. It follows that when a process takes the token to append a block it can only use the LRC communication abstraction.

\begin{theorem}\label{t:noForkBeStrong}
	There does not exist an implementation of $\mathfrak{R}(\text{BT-ADT}_{ SC}, \Theta)$ with  $\Theta\neq \Theta_{F,k=1}$ that uses a LRC primitive and generates histories satisfying the BT Strong consistency.
\end{theorem}

\begin{proof}
	Let us assume that there exist a BT-ADT implementation that satisfies the BT Strong consistency criterion refined with a $\Theta$-ADT different from $\Theta_{F,k=1}$, which implies that forks in the $bt$ can occur. Let us now construct the following history $H$ generated by the system execution at two correct processes $i$ and $j$. At the beginning $bt_i=bt_j=b_0$. 
	At the same time instant $t_0$ both processes invoke {\sf append}$(b_1)$ and {\sf append}$(b_2)$ operations respectively and $b_1, b_2 \in \mathcal{B'}$. By definition, the {\sf append}$()$ operation applies a selection function $f \in \mathcal{F}$ to select the block from the BlockTree to which the new block has to be appended, in this case such block is $f(bt_i)=f(bt_j)=f(b_0)=b_0$. 
	By construction, $b_i, b_j \in \mathcal{B'}$, let us assume that a fork occurs and both  {\sf append}$()$ operations take place and update events are triggered. Since an LRC primitive is used, each update is sent to the other processes. Since synchronous channels are employed, then by time $t_0+\delta$ the update events are delivered by $i$ and $j$.
	Let us consider that $H$ contains the following ordered events: $update_i(b_0,b_j) \mapsto update_i(b_0,b_i)$ and $update_j(b_0,b_i) \mapsto update_j(b_0,b_j)$. It follows that at a time instant $t<t_0+\delta$ it can occur that $bt_i=b_0 ^\frown b_j$ and $bt_j=b_0 ^\frown b_i$. Let us finally assume that at time $t$ both $i$ and $j$ issue a {\sf read}$()$ operation. By definition it returns the result of the selection function $f$ to the BlockTree. For both processes the BlockTree is a blockchain, thus the {\sf read}$()$ operations returns $b_0 ^\frown b_j$ at $i$ and $b_0 ^\frown b_i$ at $j$ violating the Strong Prefix property leading to a contradiction. Thus, there no exists an implementation of a BT-ADT refined with a $\Theta$-ADT different from $\Theta_{F,k=1}$ that generates histories satisfying the BT Strong consistency even in a fault-free environment.
\end{proof}

From Theorem \ref{t:noForkBeStrong} the next Corollary follows.
\begin{corollary}
	$\Theta_{F,k=1}$ is necessary for any implementation of any $\mathfrak{R}(\text{BT-ADT}_{ SC}, \Theta)$ that generates histories satisfying the BT Strong consistency.
\end{corollary}

Thanks to Theorem \ref{t:consensusOracleReduction} the next Corollary also follows.

\begin{corollary}
	Consensus is necessary for any implementation of a BT-ADT that generates histories satisfying the BT Strong consistency.
\end{corollary}


As direct implication of the Theorem \ref{t:noForkBeStrong} we can eliminate from the hierarchy in Figure \ref{fig:BTHierarchy} both $\mathfrak{R}(\text{BT-ADT}_{ SC}, \Theta_{P})$ and $\mathfrak{R}(\text{BT-ADT}_{ SC}, \Theta_{F,k>1})$, since in both cases the $\Theta$-ADT employed allows forks, thus such enriched ADTs can not generate histories that satisfies the BT Strong consistency criterion. The resulting hierarchy is depicted in Figure \ref{fig:BTHierarchyPurged}.

%
%
%
%
%
%

\begin{figure}
	\centering
	\begin{tikzpicture}[scale=0.9]
	
	\node[draw](1) at (0,2) {$\mathfrak{R}(\text{BT-ADT}_{SC}, \Theta_{F,k=1})$};
	\node[draw](2) at (-2,-2) {$\mathfrak{R}(\text{BT-ADT}_{ E C}, \Theta_{F,k>1})$};
	
	\node[draw, gray!40](3) at (2.5,-1) {$\mathfrak{R}(\text{BT-ADT}_{ SC}, \Theta_{P})$};
	\node[draw,  gray!40](5) at (-2,0) { $\mathfrak{R}(\text{BT-ADT}_{ SC}, \Theta_{F,k>1})$};
	\node[draw](4) at (0,-4) {$\mathfrak{R}(\text{BT-ADT}_{ E C}, \Theta_{P})$};
	
	
	\draw[->]  (2) edge (4) (5) edge (2) 	(1) edge (3)
	(1) edge (5)	(3) edge (4) ;
	
	
	\node[] () at (-2, 1.3) {\scriptsize Theorem \ref{t:oracolsHierarchyF}};
	\node[] () at (-2.5, 0.7) {\scriptsize Theorem \ref{t:noForkBeStrong}};
	\node[] () at (2.5, 1) {\scriptsize Theorem \ref{t:oracolsHierarchy}};
	\node[] () at (3, 0) {\scriptsize Theorem \ref{t:noForkBeStrong}};	
	\node[] () at (-2.2, -3) {\scriptsize Theorem \ref{t:oracolsHierarchy}};
	\node[] () at (-3.2, -1) {\scriptsize Corollary \ref{c:oracolsHierarchySameOracle}};
	\node[] () at (3, -2.5) {\scriptsize Corollary \ref{c:oracolsHierarchySameOracle}};			
	\end{tikzpicture}
	\caption{$\mathfrak{R}(\text{BT-ADT}, \Theta)$ Hierarchy. In gray the combinations impossible in a message-passing system}\label{fig:BTHierarchyPurged}
\end{figure}
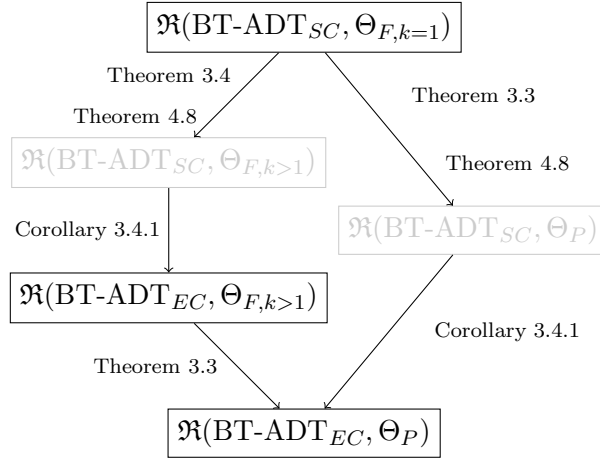

%

\section{Mapping with existing Blockchain-like systems}
\label{sec:mapp-with-exist}

This section completes this work by illustrating the mapping between different existing systems and the specifications and abstractions presented in this paper.  The following table summarizes the mapping between different existing systems and these abstractions. 
More details are given in the following sections. In those sections we refer to a permissionless system as a system where the cardinality of the process set is not a-priori known and each process can read and append into the blockchain. When we do not consider permissionless systems we explicitly state the differences.


\begin{table}[h]
	\centering
	\caption{Mapping of existing systems. Each of these systems assumes at least a light reliable communication.}
	\label{tab:mapping}
	\begin{tabular}{llllll}
		References  &  Refinement \\
		\hline
		Bitcoin~\cite{bitcoin} & \(\mathfrak{R}(BT$-$ADT_{E C}, \Theta_P)\) \\
		Ethereum~\cite{Ethereum}  &  \(\mathfrak{R}(BT$-$ADT_{E C}, \Theta_P)\) \\
		Algorand~\cite{gilad2017algorand}  &  \(\mathfrak{R}(BT$-$ADT_{S C}, \Theta_{F,k=1})\) $SC$ w.h.p      \\
		ByzCoin~\cite{BizCoin}   &  \(\mathfrak{R}(BT$-$ADT_{S C}, \Theta_{F,k=1})\) \\
		PeerCensus~\cite{PeerCensus} & \(\mathfrak{R}(BT$-$ADT_{S C}, \Theta_{F,k=1})\)\\
		Redbelly~\cite{redbelly17}    &  \(\mathfrak{R}(BT$-$ADT_{S C},\Theta_{F,k=1})\) \\
		Hyperledger~\cite{HyperLedger} &  \(\mathfrak{R}(BT$-$ADT_{S C}, \Theta_{F,k=1})\)  \\
	\end{tabular}
\end{table}

\subsection{Bitcoin}
Bitcoin~\cite{bitcoin} is the pioneer of blockchain systems. Any process \(p\in V\) is allowed to {\sf read} the BlockTree and {\sf append} blocks to the BlockTree. Processes are characterized by their computational power represented by \(\alpha_p\), normalized as \(\sum_{p\in V}\alpha_p = 1\). Processes communicate through reliable FIFO authenticated channels (implemented with TCP), which models a partially synchronous setting~\cite{jacm88}.  Valid blocks are flooded in the system. The  {\sf getToken} operation is implemented by a proof-of-work mechanism. The {\sf consumeToken} operation returns true for all valid blocks, thus there is no bounds on the number of consumed tokens. Thus Bitcoin implements a Prodigal Oracle.  The  \(f\) {\sf selects} returns the blockchain which has required the most computational work, guaranteeing that concurrent blocks can only refer to  the most recently appended blocks of the blockchain returned by a read() operation. Garay and al~\cite{GarayKL15} have shown, under a synchronous environment assumption, that Bitcoin ensures Eventual consistency criteria. The same conclusion applies as well for the FruitChain protocol \cite{PS17}, which proposes a protocol similar to BitCoin except for the rewarding mechanism.

\subsection{Ethereum}
Ethereum~\cite{Ethereum} is a permissionless blockchain. Processes are characterized by their merit parameter represented by \(\alpha_p\) (once normalized as \(\sum_{p\in V}\alpha_p = 1\)). Contrarily to Bitcoin, where this merit parameter is representative of a computational power, that is this ability to quickly compute hash functions, in Ethereum this merit is bounded by the ability to move data in memory.  This proof-of-work mechanism is especially designed for commodity hardware. Any process \(p\in V\) is allowed to {\sf read} the BlockTree and {\sf append} blocks to the BlockTree.  Processes communicate through reliable FIFO authenticated channels (implemented with TCP), which models a partially synchronous setting~\cite{jacm88}.  Valid blocks are flooded in the system. The  {\sf getToken} operation is implemented by a proof-of-work mechanism. The {\sf consumeToken} operation returns true for all valid blocks, thus there is no bounds on the number of consumed tokens. Thus Ethereum implements a Prodigal Oracle. The  \(f\) {\sf selects} returns the blockchain which has required the most work (see Section~10 of~\cite{Ethereum}), guaranteeing that concurrent blocks can only refer to  the most recently appended blocks of the blockchain returned by a {\sf read()} operation. This function is implemented through GHOST algorithm~\cite{GHOST}. Kiayias has shown~\cite{KiayiasP16}, under a synchronous environment assumption, that GHOST protocol enjoys both common prefix and chain growth properties. Ethereum thus ensures the Eventual consistency criteria.

\subsection{ByzCoin}
ByzCoin~\cite{BizCoin} is a permissionless blockchain. Processes are characterized by their computational power represented by \(\alpha_p\) (once normalized as \(\sum_{p\in V}\alpha_p = 1\)). Byzcoin assumes a semi synchronous environment, that is,
in every period of length \(b\) there must be a strongly synchronous period of length \(s < b\). The block creation process is separated from  the transaction validation one. The former one is realized by  a proof-of-work mechanism (similar to the Bitcoin's one), and the latter one is achieved by a Byzantine tolerant algorithm (i.e., a variant of PBFT~\cite{pbft}) which  creates micro blocks made of  transactions. 

The  {\sf getToken} operation is implemented by a proof-of-work mechanism.
Due to the PoW mechanism, several key blocks can be concurrently created. The {\sf consumeToken} operation guarantees that during  the   synchronous periods of the semi-synchronous setting  (those synchronous periods ensure  that everyone receives all the concurrent key blocks in a short period of time), a single key block  will be appended to the BlockTree by relying on a deterministic function \(f\) which {\sf selects} the key block whose digest (fingerprint) has the smallest least significant bits among the concurrent key blocks. Under those assumptions, Byzcoin is an implementation of a strongly consistent BlockTree composed with a Frugal Oracle, with $k=1$. 

Note that transactions do not belong to key blocks but to microblocks which are created by  a variant of PBFT where \emph{(i)} the committee members are the miners of the last $w$ appended key blocks in the BlockTree  as returned by a read() operation; \emph{(ii)}  each committee member receives a voting share for each block it has created blocks among these $w$ ones, and \emph{(iii)}  committee members are  organized on a tree rooted at the leader, and \emph{(iv)} this  leader is the process that invoked the last  successful {\sf consumeToken} operation.

\subsection{Algorand}

Algorand~\cite{gilad2017algorand} is an algorithm dedicated to permisionless blockchains. 
Users are characterized by the quantity of coins (stake) they own, represented by \(\alpha_p\) once normalized as \(\sum_{p\in V}\alpha_p = 1\). Algorand  assumes a synchronous setting (rounds) in order to ensure that \emph{(i)} with overwhelming probability all users agree on the same transactions (safety property) and \emph{(ii)}  new transactions are added to the blockchain (liveness property). Note that safety holds even in a semi synchronous environment.  Users  communicate among themselves through reliable communication channels (implemented via TCP).
Algorand algorithm relies on two main ingredients: a cryptographic sortition and a variant of a Byzantine agreement algorithm.
The cryptographic sortition implements the {\sf getToken} operation by selecting the block proposer.
This is achieved by selecting at random  a committee (that is a small fraction of users weighed by their currency balance \(\alpha_p\), which boils down to a proof-of-stake mechanism) and providing them a random priority, so that with high probability, the highest priority committee member will be in charge of proposing the new block for the current round. 
The variant of Byzantine agreement algorithm BA* implements the {\sf consumeToken} operation, that is the commitment to append this new valid block in the blockchain. 
BA* guarantees that in a favorable environment (strongly synchronous environment augmented with synchronized clocks), if all honest participants have received the same valid block, then this block will be appended to the blockchain (see Lemma 2~\cite{Algorand2017-techreport}). On the other hand, if there is no agreement on that block (because the highest priority committee member is malicious or the network is not strongly synchronous), then BA* may create forks with probability less than $10^{-7}$ (Theorem 2~\cite{Algorand2017-techreport}). This makes  Algorand
 a probabilistic implementation of a strongly consistent BlockTree composed with a Frugal Oracle, with $k=1$.

\subsection{PeerCensus}
\label{sec:peercensus}

PeerCensus~\cite{PeerCensus}  is a permissionless blockchain. Processes are characterized by their computational power represented by \(\alpha_p\) (once normalized as \(\sum_{p\in V}\alpha_p = 1\)). PeerCensus  assumes a semi synchronous environment, that is,
in every period of length \(b\) there must be a strongly synchronous period of length \(s < b\).  PeerCensus is not  strictly speaking a blockchain-based algorithm (as Bitcoin or Byzcoin), in the sense that it does not store a sequence of application transactions, but provides a secure and fully distributed timestamping service. This service is implemented by a 
dynamic Byzantine tolerant consensus algorithm which tracks the committee members of the consensus algorithm through the creation of chained key blocks. The {\sf getToken} operation is implemented by a proof-of-work mechanism, and the  {\sf consumeToken} operation, implemented by the Byzantine consensus, commits a single key block among the concurrent ones, that is returns true for a single token, as long as no more than a $1/3$ of the committees members are Byzantine (secure state). Theorem 1~\cite{PeerCensus} states that the secure state is reachable with high probability if the computational power owned by the adversary, \(\alpha_A\), is less than $1/3$. Thus under these assumptions PeerCensus  implements a strongly consistent BlockTree composed with a Frugal Oracle, with $k=1$.   Note however that in~\cite{ALLS16} the authors have analyzed the probability that PeerCensus reaches a secure  state by examing the composition of successive quorums, and have shown that this probability is decreasing as a function of \(\alpha_A\). For instance, if \(\alpha_A = 1/4\), then the probability that PeerCensus reaches a secure state is only equal to $1/3$.

\subsection{Red Belly}
Red Belly~\cite{redbelly17} is a consortium blockchain, meaning that any process \(p \in V\) is allowed to {\sf read} the BlockTree but a predefined subset \(M \subseteq V\) of processes are allowed to {\sf append} blocks. Each process \(p \in M\) as a merit parameter set to \(\alpha_p = 1/|M|\) while each process \(p \in V \setminus M\) has a merit parameter \(\alpha_p = 0\). Processes are asynchronous (i.e.,  there is  no assumption on their respective computational speed) and are connected with partially synchronous~\cite{jacm88} (i.e., messages are delivered in unknown but finite time), reliable and authenticated communication channels. 
Each process \(p \in M\) can invoke the {\sf getToken} operation with their new block and will receive a token. The  {\sf consumeToken} operation, implemented by a Byzantine consensus algorithm run by all the processes in \(V\), returns true for the uniquely decided  block. 
Thus Red Belly  BlockTree  contains a unique blockchain, meaning that the selection function \(f\) is the trivial projection function from \(\mathcal{BT} \mapsto \mathcal{BC}\) which associates to the BT-ADT its unique existing chain of the BlockTree.
 As a consequence Red Belly relies on a Frugal Oracle with $k=1$, and by the properties of Byzantine agreement implements a  strongly consistent BlockTree (see Theorem 3~\cite{redbelly17}). 
 
 \subsection{HyperLedger Fabric}
 HyperLedger Fabric~\cite{HyperLedger} is a system allowing to deploy and operate persmissioned blockchains. Any process \(p\in V\) is allowed to {\sf read} the BlockTree, however, only a subset of \(M \subseteq V\) is allowed to {\sf append} blocks to the BlockTree. Every process of \(M\) has the same merit parameter \(\alpha_M = 1/|M|\) while processes of \(V \setminus M\) have a null merit parameter. 
HyperLedger Fabric assumes eventual synchrony and reliable channels. Transactions are executed by a dedicated set of processes called endorsers. Executed transactions are then ordered through atomic broadcast primitive so as to gather them into a block. HyperLedger Fabric relies on a leader election to determine which process will generate the next block. Transactions are appended in a block until a stop condition is met. A stop condition refers either on a maximal number of transactions in a block or a maximal elapsed time since the first transaction included in the block. The block is then broadcasted and a new block is created to gather new incoming transactions.
By construction, HyperLedger Fabric ensures that a unique token (\(k=1\)) is consumed, thus HyperLedger Fabric implement a strongly consistent BlockTree.

\section{Conclusions and Future Work}
The paper presented an extended formal specification of blockchains and derived interesting conclusion on their implementability. 
Let us note that the presented work is intended to provide the groundwork  for  the  construction of a sound hierarchy of blockchain abstractions and correct implementations.  Future work will focus on several open issues, such as the solvability of Eventual Prefix in message-passing, the synchronization power of other oracle models, and fairness properties for oracles. 

\section*{Acknowledgment}
We are grateful to Mathieu Perrin and anonymous reviewers for their insightful comments on a previous version of the current paper.



\newpage
\bibliographystyle{plain}
\bibliography{btadt}

\end{document}